\documentclass[10pt,reqno,oneside]{amsart}

\usepackage[T1]{fontenc} \usepackage{amsmath} \usepackage{amsfonts} \usepackage{amsthm} \usepackage{amssymb} \usepackage{graphicx} \graphicspath{{figures/}} \usepackage{epstopdf}
\usepackage{hyperref} \hypersetup{hidelinks} \usepackage{cite} \usepackage{calc} \usepackage{geometry}

\usepackage{tikz} \usepackage{tikz-cd}
\usetikzlibrary{%
  matrix,%
  calc,%
  arrows%
}

\definecolor{curvecol}{HTML}{1F3B73} \definecolor{altcol}{HTML}{A02020} \definecolor{guidecol}{HTML}{C8C8C8} \definecolor{eqcol}{HTML}{8A8A8A} \tikzset{ guide/.style      ={draw=guidecol, line width=0.4pt}, guidefar/.style   ={draw=guidecol, line width=0.35pt, opacity=0.40}, equator/.style    ={draw=eqcol, line width=0.55pt, dash pattern=on 2.4pt off 1.8pt}, equatorfar/.style ={draw=eqcol, line width=0.5pt, opacity=0.40, dash pattern=on 2.4pt off 1.8pt}, curve/.style      ={draw=curvecol, line width=0.7pt, line cap=round, line join=round}, curvefar/.style   ={draw=curvecol, line width=0.6pt, opacity=0.20, line cap=round, line join=round}, curvealt/.style   ={draw=altcol, line width=0.32pt, line cap=round, line join=round}, curvealtfar/.style={draw=altcol, line width=0.3pt, opacity=0.20, line cap=round, line join=round}, polarcap/.style   ={draw=altcol, line width=0.5pt, dash pattern=on 1.7pt off 1.3pt}, polarcapfar/.style={draw=altcol, line width=0.45pt, opacity=0.40, dash pattern=on 1.7pt off 1.3pt}, sepline/.style    ={draw=eqcol, line width=0.5pt, dash pattern=on 2pt off 2pt}, geoline/.style    ={draw=eqcol, line width=0.5pt, dash pattern=on 4pt off 2.5pt}, axis/.style       ={draw=black!55, line width=0.5pt}, } \newlength{\sphrad} 
\newlength{\panelrad} 

\newcommand{\ncd}{\newcommand} \ncd{\mrm}    {\mathrm} \ncd{\beq} {\begin{equation}} \ncd{\eeq} {\end{equation}} \ncd{\nn}{\nonumber}

\def\d{{\rm d}}  \def\iiota{\dot\iota}

     \def\Reals{ \mathbb{R}} \def\Sph{ \mathbb{S}} \def\basis[#1]{\frac{\partial}{\partial #1}} \def\dt[#1]{\frac{\d}{\d #1}} \def\GTM{\Gamma\left(TM \right)}  \def\sn{\,{\rm sn}} \def\cn{\,{\rm cn}} \def\dn{\,{\rm dn}} \def\lcan{\lambda_{\mrm{can}}}

\newtheorem{prop}{Proposition} \newtheorem{theo}{Theorem} \newtheorem{corollary}{Corollary} \newtheorem{lemma}{Lemma} \newtheorem{mydef}{Definition} \theoremstyle{remark} \newtheorem{remark}{Remark}

\begin{document}

\title[Closure of electromagnetic curves on the sphere]{A closure criterion for electromagnetic curves on the sphere in a uniform ambient field}

\author{C\'esar S. L\'opez-Monsalvo} \address{Departamento de Ciencias B\'asicas, Universidad Aut\'onoma Metropolitana -- Azcapotzalco, Avenida San Pablo 420, Colonia Nueva El Rosario, Azcapotzalco 02128, Ciudad de M\'exico, M\'exico} \email{cslm@azc.uam.mx} \urladdr{https://orcid.org/0000-0002-0378-0415}



\begin{abstract} We consider the motion of a charged test particle confined to the round unit sphere in a uniform ambient magnetic field. We use an extension of Noether's theorem for systems with magnetic forces to reduce the problem to a quadrature. The rotation number of a trajectory, the azimuth it gains over one oscillation in latitude, is a complete elliptic integral of the third kind in Legendre form. A trajectory closes if and only if that advance is a rational multiple of a full turn. The criterion holds on every level set of the two first integrals. Then, we differentiate the rotation number with respect to the half-cyclotron frequency and sign that derivative on each of the two branches into which a level set divides. The rotation number sees the charge, the mass, the field and the speed through a single dimensionless ratio, that of the half-cyclotron frequency to the speed. A closed trajectory therefore fixes a value of that ratio. We say how many values carry a given rational winding. The poles are attainable on a single level set, where the motion reduces to a pendulum. The field has no flux through the sphere, so it is globally exact. We compute the Ma\~n\'e strict critical value of the resulting exact magnetic flow. Above that value the closed trajectories are closed Reeb orbits of an explicit contact form. \end{abstract}

\maketitle

%%============================================================
\section{Introduction} \label{sec.intro}
%%============================================================

\emph{When does a charged particle confined to a sphere retrace its own path?} The system that raises the question is an old one. Place the unit sphere in a uniform ambient magnetic field, restrict that field to the surface, and let a charge move on it. Indeed, it is the simplest configuration in which a curved surface meets a uniform field, and the motion it produces is neither geodesic nor planar. Most of the trajectories precess indefinitely and fill a band between two parallels of latitude. A few of them close only for particular combinations of the charge, the mass, the field and the speed, which the geometry alone picks out of a continuum. Determining that set is the problem we address here.

Saksida formulated the motion in terms of the Neumann system and the spherical pendulum~\cite{saksida2002neumann}. Dragovi\'c, Gaji\'c and Jovanovi\'c obtained this flow as a reduction of a nonholonomic problem, proved it Liouville integrable on $T^*\Sph^2$, and integrated its equations of motion explicitly in elliptic functions~\cite{dragovic2023gyroscopic}. The trajectories were therefore known in closed form. However, in higher dimensions the same field --- the restriction to $\Sph^n$ of a constant $2$-form of $\Reals^{n+1}$ --- proved harder. Its integrability was settled only recently, first by the same authors for $n \leq 5$~\cite{dragovic2025homogeneous}, then in every dimension. Two arguments of the same date settle it, one of Bolsinov, Konyaev and Matveev reducing the flow to a degenerate Neumann system~\cite{bolsinov2025integrability}, one of Dragovi\'c, Gaji\'c and Jovanovi\'c giving a Lax representation~\cite{dragovic2025lax}. Kobtsev and Kudryavtseva classified the Liouville fibration of the rotationally invariant case on the two-sphere by its Fomenko--Zieschang invariant~\cite{kobtsev2025bifurcations}.

It is worth noting how the configuration studied here differs from the one on which the literature has been centred. On an oriented surface the area form is closed, so every constant multiple of it is a magnetic field. Barros, Romero, Cabrerizo and Fern\'andez showed that its trajectories are exactly the curves of constant geodesic curvature~\cite{barros2005gauss}. Cabrerizo later placed that characterisation within a broader classification~\cite{cabrerizo2013magnetic}. Thus, on the round sphere those are the small circles, and every one of them closes at every speed. The field a sphere inherits from a uniform ambient one carries the cosine of the polar angle against the area form instead of a constant, so closure becomes the question we answer here. In a different direction, the existence of closed magnetic geodesics on surfaces is known under hypotheses far weaker than ours~\cite{ginzburg2007periodic,usher2008floer,schneider2009alexandrov,asselle2016existence}, by arguments that settle existence while leaving the orbits themselves unnamed. Cort\'es and Cort\'es-Poza treat a charge restricted to a sphere in the field of a dipole and reach a quadrature of the same shape~\cite{cortes2013spherically}.

What decides whether a trajectory returns to its starting point is its \emph{rotation number}, the azimuth it gains while its latitude completes one oscillation. That number is an elliptic \emph{integral} rather than an elliptic function, and to the best of our knowledge it has not been computed for this system. We compute the rotation number and the closure criterion that follows.

The advance over one period of the polar angle is a complete elliptic integral of the third kind, which we reduce to Legendre form in the modulus and characteristic that the turning points fix (Theorem~\ref{theo.rotation}). A trajectory returns to its starting point exactly when that advance is a rational multiple of a full turn (Theorem~\ref{theo.closure}). Both first integrals stay free in that statement, so it reaches every trajectory the field carries. Moreover, the rotation number sees the charge, the mass, the field and the speed through a single dimensionless ratio. Hence, closed trajectories fix the value of that ratio (Proposition~\ref{prop.scaling}).

Then, we ask which windings a field can produce. The equator is a double root of the polar motion at exactly one value of the half-cyclotron frequency, which we call the equatorial value. It divides the admissible frequencies into a branch on which the band crosses the equator and a branch on which it lies within one hemisphere (Proposition~\ref{prop.separatrix}). Below the equatorial value the rotation number runs continuously from the winding of a great circle either to infinity or to the winding of the equatorial circle, according as the azimuthal integral falls below or above half the speed, so we can name exactly which rational windings occur (Theorem~\ref{theo.realise}). The derivative of the rotation number with respect to that frequency decides how many values of it carry a given winding, and we compute it in closed form as a combination of the polar period and a second moment of the motion (Theorem~\ref{theo.twist}). Below the equatorial value that combination is positive, so the rotation number increases strictly (Theorem~\ref{theo.mono}). Above it the band acquires a symmetry, an involution exchanging the turning points under which the quadrature is invariant, and the rotation number decreases strictly (Theorem~\ref{theo.mono2}). Each branch therefore carries each admissible rational winding exactly once. A winding occurs at two values of that frequency, at one, or at none, and we say which (Corollary~\ref{cor.count}). We have found no counterpart to that count in the literature.

One level set stands apart. The azimuthal velocity carries a term inversely proportional to the squared sine of the polar angle, weighted by the azimuthal first integral, so a trajectory on which that integral fails to vanish turns ever faster as it nears a pole and never arrives. Indeed, the quartic governing the polar motion makes the same statement in one line, since it takes the negative of the squared azimuthal integral at either pole (Proposition~\ref{prop.poles}). Hence the poles are attainable exactly on the set where that integral vanishes, where the reduction becomes elementary, the trajectory is a Jacobi elliptic function of a single modulus, and the closure criterion collapses to a resonance between a complete elliptic integral of the first kind and a full turn (Corollary~\ref{cor.ell0}). That set bounds the family rather than belonging to it.

Let us turn to the closed trajectories as closed Reeb orbits. The total flux of the field through the sphere vanishes, so the field is globally exact and the twisted symplectic form is exact on the whole cotangent bundle (Proposition~\ref{prop.exact}). For an exact magnetic flow on a closed surface other than the torus, Contreras, Macarini and Paternain settled when an energy level is of contact type, and it is so precisely above the Ma\~n\'e strict critical value~\cite{contreras2004periodic}. We compute that value for this field and find it to be half the mass times the square of the half-cyclotron frequency (Proposition~\ref{prop.mane}), and we exhibit the Liouville field that realises the threshold. Thus, above it the closed trajectories of Theorem~\ref{theo.closure} are explicit closed Reeb orbits, one for each rational winding the level admits (Corollary~\ref{cor.reeb}).

Our method is the one established for the Dirac monopole~\cite{lopez2025noether}, whose definition of the Lorentz endomorphism we adopt unchanged, in our own notation. A reader moving between the two therefore meets the same object. Sternberg's construction lifts the field $2$-form to phase space and twists the canonical symplectic form~\cite{sternberg1977minimal,guillemin1990symplectic}, so that the free particle flow of the twisted structure projects onto the electromagnetic curves of the base. That flow lies outside the reach of Noether's theorem. Here we use the Ikawa extension~\cite{ikawa2003hamiltonian,ikawa2007motion,ikawa2010motion}, which delivers the azimuthal first integral. Indeed, that integral is already in the literature, as the linear integral of the Liouville pair of Dragovi\'c, Gaji\'c and Jovanovi\'c~\cite{dragovic2023gyroscopic}. We re-derive it here because the Noether--Ikawa route names the geometry the rest of the manuscript needs, and the potential it produces reappears in the Ma\~n\'e value of Proposition~\ref{prop.mane}.

The manuscript is structured as follows. In Section~\ref{sec.problem} we state the electromagnetic curves problem on a Riemannian manifold and fix the compatibility convention that relates the field to the Lorentz endomorphism. In Section~\ref{sec.sphere} we set up the sphere in a uniform ambient field and derive the equations of motion. In Section~\ref{sec.symplectic} we recall Sternberg's lift and Ikawa's extension of Noether's theorem, tie the twist to the compatibility convention, and obtain the first integrals and the quadrature. In Section~\ref{sec.closed} we compute the rotation number, prove the closure criterion and treat the boundary case. In Section~\ref{sec.contact} we establish exactness, compute the Ma\~n\'e critical value and identify the closed orbits above it as closed Reeb orbits. In Section~\ref{sec.closing} we take stock of what we prove, what we assume and what remains open.

%%============================================================
\section{The electromagnetic curves problem} \label{sec.problem}
%%============================================================

Throughout the manuscript, $(M,g)$ is a Riemannian manifold and $\gamma: I \subset \Reals \longrightarrow M$ is a curve parametrised by arc length, with velocity $\dot\gamma \in TM$. Here $TM$ is the tangent bundle of $M$ and an over dot denotes Lie differentiation along the flow of $\dot\gamma$. A test particle of mass $m$ and charge $q$ follows the curve, and the field it moves in is a closed $2$-form.

\begin{mydef}[Electromagnetic field] \label{def.field} An \emph{electromagnetic field} on $(M,g)$ is a $2$-form $F \in \Lambda^2(M)$ satisfying \beq \label{eq.maxwell} \d F = 0. \eeq \end{mydef}

The condition~\eqref{eq.maxwell} is the homogeneous half of Maxwell's equations. The inhomogeneous half determines which sources produce the field, and we take the field as given. In this sense the whole construction rests on~\eqref{eq.maxwell} alone.

Let us now fix the tensor that carries the Lorentz force, following the monopole study~\cite{lopez2025noether} and the geometry of electromagnetic curves set out in~\cite{islas2021geometry}.

\begin{mydef}[Lorentz endomorphism] \label{def.endo} The \emph{Lorentz endomorphism} associated with the field $F$ is the bundle endomorphism $\Phi: TM \longrightarrow TM$ determined by the compatibility condition \beq \label{eq.compat} g\left(\Phi(\zeta),\eta\right) = F(\zeta,\eta) , \eeq for every pair of vector fields $\zeta$ and $\eta$. \end{mydef}

In a local chart the compatibility~\eqref{eq.compat} reads $\Phi^a_{\ b} = g^{ac}F_{bc}$, the field with its second index raised by the metric, which is the content of the equivalent statement $\Phi(\zeta) = g^\sharp\left(\iiota_\zeta F\right)$. Here $\iiota_\zeta$ denotes the interior product with $\zeta$, while $g^\sharp$ and $g^\flat$ are the musical isomorphisms the metric induces between vector fields and $1$-forms. The reader who prefers the opposite pairing obtains the endomorphism of the opposite sign.

\begin{mydef}[Electromagnetic curve] \label{def.emcurve} An \emph{electromagnetic curve} of the field $F$ on $(M,g)$ is a curve $\gamma$ satisfying \beq \label{eq.eom} m \nabla_{\dot\gamma}\dot\gamma = q\, \Phi(\dot\gamma), \eeq where $m$ and $q$ are the mass and the charge of the test particle, and $\nabla$ is the Levi-Civita connection of $g$. \end{mydef}

The law~\eqref{eq.eom} identifies the force with the change of momentum and adopts the Levi-Civita connection of $g$. What that identification fixes of the pair $(g,\nabla)$, and what it leaves free, we have examined elsewhere~\cite{lopez2026geometry}. Here we take both as given.

The compatibility~\eqref{eq.compat} has two immediate consequences. The force is orthogonal to the velocity, since \beq \label{eq.orth} g\left(\Phi(\dot\gamma),\dot\gamma\right) = F(\dot\gamma,\dot\gamma) = 0, \eeq by the antisymmetry of $F$, which makes the ordering of the arguments in~\eqref{eq.compat} immaterial here. Consequently the motion preserves the speed, \beq \label{eq.speed} \dt[\tau] \vert \dot\gamma \vert^2 = 2\, g\!\left(\nabla_{\dot\gamma}\dot\gamma, \dot\gamma\right) = \frac{2q}{m}\, g\!\left(\Phi(\dot\gamma),\dot\gamma\right) = 0, \eeq where we have used that the Levi-Civita connection is metric compatible. The dynamics itself thus enforces the arc-length parametrisation assumed at the outset, and we write $\vert\dot\gamma\vert = v$ for the constant speed.

Electromagnetic curves are not the geodesics of any affine connection on $M$ --- this is Proposition~$2.1$ of~\cite{barros2005gauss} --- and the linear dependence of the Lorentz force on the velocity makes them irreversible. Indeed, the curve traversed backwards solves the equation of the opposite charge. That irreversibility is a feature of the problem rather than of the description.

%%============================================================
\section{The sphere in a uniform ambient field} \label{sec.sphere}
%%============================================================

Let us now specialise to the geometry of the problem. Let $\Sph^2$ be the unit sphere carrying its round metric in spherical coordinates, \beq \label{eq.metric} g = \d\theta \otimes \d\theta + \sin^2\theta\ \d\varphi \otimes \d\varphi, \eeq with $\theta \in (0,\pi)$ the polar angle measured from the north pole and $\varphi \in [0,2\pi)$ the azimuth.

\begin{mydef}[Restricted uniform field] \label{def.restricted} Let $\iota: \Sph^2 \hookrightarrow \Reals^3$ be the standard embedding and let $F_{\Reals^3} = B\ \d X \wedge \d Y$ be the field of a uniform ambient magnetic field of strength $B$ along the polar axis. The \emph{restricted uniform field} on the sphere is its pullback, \beq \label{eq.field} F = \iota^* F_{\Reals^3} = B \sin\theta \cos\theta\ \d\theta \wedge \d\varphi. \eeq \end{mydef}

We obtain the pullback by a direct computation from the embedding $X = \sin\theta\cos\varphi$, $Y = \sin\theta\sin\varphi$ and $Z = \cos\theta$, the Jacobian of the first two with respect to $(\theta,\varphi)$ being $\sin\theta\cos\theta$. Indeed, the field~\eqref{eq.field} is closed for the trivial reason that every $2$-form on a surface is closed, and it meets Definition~\ref{def.field} on that ground alone.

It is worth noting how the restricted field stands to the area form. The area form of the round metric~\eqref{eq.metric} is $\mrm{Vol}_g = \sin\theta\ \d\theta \wedge \d\varphi$, so the restricted field~\eqref{eq.field} reads \beq \label{eq.notarea} F = B\cos\theta\ \mrm{Vol}_g , \eeq a \emph{function} multiple of the area form, varying with the polar angle. The Gauss--Landau--Hall field of~\cite{barros2005gauss,cabrerizo2013magnetic} is the case in which that function is constant, where $\Phi$ becomes a constant multiple of the rotation by a quarter turn. Here the multiplier vanishes on the equator and changes sign across it.

Throughout the remainder of the manuscript we abbreviate the half-cyclotron frequency \beq \label{eq.alpha} \alpha := \frac{qB}{2m}, \eeq which carries the dimensions of an inverse time and is the only combination of the charge, the mass and the field strength the dynamics uses.

\begin{prop}[Equations of motion] \label{prop.eom} The electromagnetic curves of the restricted uniform field~\eqref{eq.field} on the round sphere~\eqref{eq.metric} are the solutions of \beq \label{eq.eomtheta} \ddot\theta = \sin\theta\cos\theta\ \dot\varphi \left( \dot\varphi - 2\alpha \right), \eeq \beq \label{eq.eomphi} \ddot\varphi = -2\cot\theta\ \dot\theta \left( \dot\varphi - \alpha \right). \eeq \end{prop}

\begin{proof} The non-vanishing Christoffel symbols of the round metric~\eqref{eq.metric} are $\Gamma^\theta_{\ \varphi\varphi} = -\sin\theta\cos\theta$ and $\Gamma^\varphi_{\ \theta\varphi} = \Gamma^\varphi_{\ \varphi\theta} = \cot\theta$, so the acceleration of a curve $\gamma(\tau) = \left(\theta(\tau),\varphi(\tau)\right)$ has components \beq \label{eq.acc} \left(\nabla_{\dot\gamma}\dot\gamma\right)^\theta = \ddot\theta - \sin\theta\cos\theta\ \dot\varphi^2 \quad \text{and} \quad \left(\nabla_{\dot\gamma}\dot\gamma\right)^\varphi = \ddot\varphi + 2\cot\theta\ \dot\theta\ \dot\varphi. \eeq The Lorentz endomorphism of the field~\eqref{eq.field} determined by the compatibility~\eqref{eq.compat} is \beq \label{eq.endo} \Phi = -B\sin\theta\cos\theta\ \basis[\theta]\otimes \d\varphi + B\cot\theta\ \basis[\varphi]\otimes \d\theta , \eeq as one verifies by evaluating both sides of~\eqref{eq.compat} on the coordinate basis. Applying~\eqref{eq.endo} to the velocity gives \beq \label{eq.force} \Phi(\dot\gamma)^\theta = -B\sin\theta\cos\theta\ \dot\varphi \quad \text{and} \quad \Phi(\dot\gamma)^\varphi = B\cot\theta\ \dot\theta, \eeq Substituting the acceleration~\eqref{eq.acc} and the force~\eqref{eq.force} into the equation of motion~\eqref{eq.eom} yields the polar equation~\eqref{eq.eomtheta} and the azimuthal one~\eqref{eq.eomphi}, once we write $qB/m = 2\alpha$ with the abbreviation~\eqref{eq.alpha}. \end{proof}

Let us dwell on the coefficients of the system~\eqref{eq.eomtheta}--\eqref{eq.eomphi}, for they carry more weight than the derivation alone suggests. The sign of the field term in the polar equation and the factor of two in the azimuthal one are precisely what the conservation of the azimuthal angular momentum reduced by $\alpha\sin^2\theta$ demands. That integral carries the whole reduction.

%%============================================================
\section{The symplectic lift and its first integrals} \label{sec.symplectic}
%%============================================================

\emph{Where does the symmetry of the problem hide?} Integrating the equations of motion~\eqref{eq.eomtheta}--\eqref{eq.eomphi} directly keeps it out of sight. Let us lift them to phase space instead, where it stands plainly.

The electromagnetic field determines a symplectic structure~\cite{sternberg1977minimal,guillemin1990symplectic}. On the cotangent bundle $\pi: T^*M \longrightarrow M$ with tautological $1$-form $\lcan$ and canonical symplectic form $\omega_0 = -\d\lcan$, the field furnishes the twisted form \beq \label{eq.twisted} \omega_F = \omega_0 + q\, \pi^*(F), \eeq which is closed because $F$ is, by the closedness condition~\eqref{eq.maxwell}, and non-degenerate because the twist is a pullback from the base. The Hamiltonian flow of the \emph{free} particle Hamiltonian \beq \label{eq.ham} h = \frac{1}{2m}\, g^{-1}(p,p) , \eeq with respect to $\omega_F$ projects onto the electromagnetic curves of~\eqref{eq.eom}. Here $p$ is the momentum, the fibre coordinate of the cotangent bundle, and $g^{-1}$ the inverse metric that pairs it with itself. In this sense the dynamics we treat is free. The twist absorbs the field into the geometry of phase space.

The projected motion is not geodesic, so Noether's theorem is unavailable to us. Nevertheless, Ikawa has extended it~\cite{ikawa2003hamiltonian,ikawa2007motion,ikawa2010motion}. We recall that extension in the form the monopole study used~\cite{lopez2025noether}. Let \beq \label{eq.killing} \Xi = \left\{ \xi \in \GTM \ \vert \ \pounds_\xi g = 0 \quad \text{and} \quad \pounds_\xi F = 0 \right\} , \eeq be the set of generators of the symmetries of both the metric and the field. Here $\GTM$ is the space of smooth vector fields on $M$, the sections of its tangent bundle, and $\pounds_\xi$ is the Lie derivative along $\xi$. The field is closed, so $\pounds_\xi F = \d\, \iiota_\xi F$ for every $\xi \in \Xi$. Indeed, the contraction $\iiota_\xi F$ is then closed, admitting, at least locally, a potential \beq \label{eq.psi} \iiota_\xi F = \d \psi_\xi . \eeq The first integral Ikawa's theorem associates with $\xi$ combines the momentum the generator lowers with that potential, \beq \label{eq.noether} \lambda_\xi = \iiota_{\dot\gamma}\, g^\flat(\xi) + \frac{q}{m}\, \psi_\xi . \eeq

\begin{prop}[The azimuthal first integral] \label{prop.first} The azimuthal vector field $\partial/\partial\varphi$ belongs to the symmetry set $\Xi$ of~\eqref{eq.killing} for the round metric~\eqref{eq.metric} and the restricted uniform field~\eqref{eq.field}, and the first integral~\eqref{eq.noether} it generates is \beq \label{eq.ell} \ell = \sin^2\theta \left( \dot\varphi - \alpha \right). \eeq \end{prop}

\begin{proof} Neither the metric~\eqref{eq.metric} nor the field~\eqref{eq.field} has any component depending on $\varphi$, so both Lie derivatives in~\eqref{eq.killing} vanish and $\partial/\partial\varphi \in \Xi$. Contracting the field~\eqref{eq.field} with the generator gives \beq \label{eq.contract} \iiota_{\partial_\varphi} F = -B\sin\theta\cos\theta\ \d\theta = \d\left( -\frac{B}{2}\sin^2\theta \right), \eeq so that $\psi_{\partial_\varphi} = -B \sin^2\theta/2$ is a potential in the sense of~\eqref{eq.psi}, defined globally on the sphere. The metric dual of the generator is $g^\flat(\partial_\varphi) = \sin^2\theta\ \d\varphi$, whose contraction with the velocity is $\sin^2\theta\ \dot\varphi$. Substituting both into the Noether--Ikawa integral~\eqref{eq.noether} gives $\lambda_{\partial_\varphi} = \sin^2\theta\ \dot\varphi - (qB/2m)\sin^2\theta$, which is the asserted expression~\eqref{eq.ell} written with the abbreviation~\eqref{eq.alpha}. We may check its conservation directly. Differentiating the first integral~\eqref{eq.ell} along the curve gives \beq \label{eq.ellcheck} \dt[\tau] \ell = 2\sin\theta\cos\theta\ \dot\theta \left( \dot\varphi - \alpha \right) + \sin^2\theta\ \ddot\varphi , \eeq The azimuthal equation of motion~\eqref{eq.eomphi} states that $\sin^2\theta\ \ddot\varphi = -2\sin\theta\cos\theta\ \dot\theta \left(\dot\varphi-\alpha\right)$, which cancels the first term of~\eqref{eq.ellcheck} identically. \end{proof}

The integral~\eqref{eq.ell} is the azimuthal angular momentum reduced by $\alpha \sin^2\theta$, the term the field contributes. It is the linear integral of the Liouville pair that Section~\ref{sec.intro} attributes~\cite{dragovic2023gyroscopic}. Its counterpart in the monopole problem is the Poincar\'e integral~\cite{poincare1896remarques}, where the shift is proportional to $\cos\vartheta$ instead of $\sin^2\vartheta$, since each shift is the potential the Noether--Ikawa construction assigns to its own field~\cite{lopez2025noether}. Together with the constant speed of~\eqref{eq.speed} it exhausts the problem.

\begin{remark}[Three signs that are one convention] \label{rem.sign} Definition~\ref{def.endo} fixes the sign of the twist in~\eqref{eq.twisted}. Solving Hamilton's equations $\iiota_{X_h}\omega_F = \d h$ for the twisted form $\omega_0 + s\,q\,\pi^*(F)$ with $s = \pm 1$, we find that the projected motion obeys \beq \label{eq.signcheck} \ddot\theta = \sin\theta\cos\theta\ \dot\varphi \left( \dot\varphi - 2 s\,\alpha \right), \eeq so that the two values of $s$ describe the two charges, of which only $s = +1$ returns the polar equation~\eqref{eq.eomtheta} that Definition~\ref{def.endo} produced. Reversing the ordering of the arguments in the compatibility~\eqref{eq.compat} conjugates the charge in the same way, and so does reversing the sign of the potential term in the Noether--Ikawa integral~\eqref{eq.noether}, since the quantity~\eqref{eq.ell} is conserved for one of those signs and not for the other. The three reversals are therefore a single convention. Performing all of them describes the same particle, and performing one of them describes the opposite charge while appearing to describe the stated one. We fix that convention at Definition~\ref{def.endo}, which is the definition of the monopole study~\cite{lopez2025noether}, and let the twist and the first integral follow from it. \end{remark}

\begin{prop}[Quadrature] \label{prop.quadrature} The system~\eqref{eq.eomtheta}--\eqref{eq.eomphi} is Liouville integrable. On the level set $\left\{ \ell,\ v \right\}$ the polar angle obeys the quadrature \beq \label{eq.quad} \dot\theta^2 = v^2 - \frac{\left(\ell + \alpha \sin^2\theta\right)^2} {\sin^2\theta} . \eeq In the variable $x = \cos\theta$ this is the even quartic \beq \label{eq.quartic} \dot x^2 = -\alpha^2 x^4 + \left[ 2\alpha(\ell + \alpha) - v^2 \right] x^2 + \left[ v^2 - (\ell+\alpha)^2 \right] . \eeq \end{prop}

\begin{proof} The configuration space has two degrees of freedom, while the two integrals~\eqref{eq.ell} and~\eqref{eq.speed} are in involution and functionally independent away from the poles, so the system is Liouville integrable, a conclusion Section~\ref{sec.intro} attributes and sets in its wider setting~\cite{dragovic2023gyroscopic,bolsinov2025integrability}. We record the reduction here for the quadrature it produces. Solving the azimuthal integral~\eqref{eq.ell} for the azimuthal velocity gives \beq \label{eq.phidot} \dot\varphi = \frac{\ell}{\sin^2\theta} + \alpha , \eeq Substituting~\eqref{eq.phidot} into the constant speed $v^2 = \dot\theta^2 + \sin^2\theta\ \dot\varphi^2$ yields the quadrature~\eqref{eq.quad}. Differentiating $x = \cos\theta$ gives $\dot x^2 = \sin^2\theta\ \dot\theta^2$; multiplying~\eqref{eq.quad} by $\sin^2\theta = 1 - x^2$ and expanding produces the quartic~\eqref{eq.quartic}. \end{proof}

The quartic~\eqref{eq.quartic} is even, of the same type as the one governing the spherical pendulum. We write $P$ for the polynomial on its right-hand side, so that $\dot x^2 = P(x)$. Thus the general solution is a Jacobi elliptic function of the polar angle. The root pattern of the quartic carries over the qualitative structure of the polar motion. The azimuth follows from~\eqref{eq.phidot} by an elliptic integral of the third kind. Over one polar period the motion runs between two turning points, so that integral is a \emph{complete} one.

%%============================================================
\section{Closed orbits at vanishing azimuthal integral} \label{sec.closed}
%%============================================================

\emph{Which of these trajectories close?} Let us settle that question here. The answer turns on a single number attached to each level set, the azimuth a trajectory gains between successive returns to a turning point in latitude, and the criterion is that this number be commensurable with a full turn.

Throughout this section we take $\alpha > 0$. This costs us no generality. Indeed, the reflection $\varphi \mapsto -\varphi$ sends the pair $(\alpha,\ell)$ to $(-\alpha,-\ell)$, carrying each trajectory to the mirror trajectory of the opposite charge. It reverses the sense of the azimuthal drift, and leaves the polar motion, the periods and the rotation number up to sign unaltered.

\begin{prop} \label{prop.poles} The polar quartic~\eqref{eq.quartic} takes the value \beq \label{eq.Ppole} P(\pm 1) = -\ell^2 \eeq at either pole. A trajectory therefore reaches a pole if and only if $\ell = 0$, and every trajectory with $\ell \neq 0$ is confined to a band between two parallels of latitude. \end{prop}

\begin{proof} Setting $x = \pm 1$ in the quartic~\eqref{eq.quartic} and writing $L = \ell + \alpha$ gives $-\alpha^2 + (2\alpha L - v^2) + v^2 - L^2 = -(L-\alpha)^2 = -\ell^2$, which is the assertion~\eqref{eq.Ppole}. Indeed, since $\dot x^2 = P(x)$ the motion occupies $\left\{ P \geq 0 \right\}$, so $x = \pm 1$ lies in its closure exactly when $\ell = 0$. Moreover, the azimuthal velocity~\eqref{eq.phidot} says the same, since it carries the term $\ell/\sin^2\theta$ and thus diverges at a pole for $\ell \neq 0$. \end{proof}

The two descriptions agree. A trajectory with $\ell \neq 0$ turns ever faster in azimuth as it approaches a pole, so that the polar motion turns back before it arrives. In this sense the level set $\ell = 0$ marks the edge of the family. The elementary reduction it supports is a degenerate case of the general one.

\begin{theo} \label{theo.ell0} On the level set $\ell = 0$ the azimuthal velocity is constant, \beq \label{eq.phifrozen} \dot\varphi = \alpha , \qquad \text{so that} \qquad \varphi(\tau) = \alpha \tau + \varphi_0 , \eeq and the polar angle satisfies the pendulum equation \beq \label{eq.pendulum} \ddot\theta = -\alpha^2 \sin\theta\cos\theta , \qquad \dot\theta^2 = v^2 - \alpha^2 \sin^2\theta . \eeq Writing $k = \min(\alpha,v)/\max(\alpha,v)$, the solutions are, in the rotation regime $v > \alpha$, \beq \label{eq.rotation} \sin\theta(\tau) = \sn(v\tau,k) , \qquad \cos\theta(\tau) = \cn(v\tau,k) , \qquad k = \frac{\alpha}{v} , \eeq and, in the libration regime $v < \alpha$, \beq \label{eq.libration} \sin\theta(\tau) = k \sn(\alpha\tau,k) , \qquad \cos\theta(\tau) = \dn(\alpha\tau,k) , \qquad k = \frac{v}{\alpha} , \eeq where $\sn$, $\cn$ and $\dn$ are the Jacobi elliptic functions of modulus $k$. \end{theo}

\begin{proof} Setting $\ell = 0$ in the expression~\eqref{eq.phidot} for the azimuthal velocity gives $\dot\varphi = \alpha$ at once. Substituting that constant into the azimuthal equation of motion~\eqref{eq.eomphi} gives $\ddot\varphi = -2\cot\theta\ \dot\theta\ (\alpha - \alpha) = 0$, so the value is consistent and the azimuth advances uniformly as in~\eqref{eq.phifrozen}. Substituting $\dot\varphi = \alpha$ into the polar equation~\eqref{eq.eomtheta} gives $\ddot\theta = \sin\theta\cos\theta\ \alpha\,(\alpha - 2\alpha) = -\alpha^2 \sin\theta\cos\theta$, which is the pendulum equation of~\eqref{eq.pendulum}. Setting $\ell = 0$ in the quadrature~\eqref{eq.quad} gives its first integral. Indeed, the two are consistent, for we recover the pendulum equation by differentiating that first integral.

For the rotation regime, let $k = \alpha/v < 1$ and put $\sin\theta = \sn(v\tau,k)$. Then $\cos\theta = \sqrt{1 - \sn^2(v\tau,k)} = \cn(v\tau,k)$ by the first Jacobi identity. Differentiating gives $\cos\theta\ \dot\theta = v \cn(v\tau,k)\dn(v\tau,k)$, whence $\dot\theta = v \dn(v\tau,k)$. Squaring and using the second Jacobi identity $\dn^2 = 1 - k^2\sn^2$ gives \beq \label{eq.rotcheck} \dot\theta^2 = v^2 \left[ 1 - k^2 \sn^2(v\tau,k) \right] = v^2 - \alpha^2 \sin^2\theta , \eeq which is the first integral of~\eqref{eq.pendulum}. For the libration regime, let $k = v/\alpha < 1$ and put $\sin\theta = k\sn(\alpha\tau,k)$. Then $\cos\theta = \sqrt{1 - k^2\sn^2(\alpha\tau,k)} = \dn(\alpha\tau,k)$ and $\dot\theta = \alpha k \cn(\alpha\tau,k)$, so that \beq \label{eq.libcheck} \dot\theta^2 = \alpha^2 k^2 \left[ 1 - \sn^2(\alpha\tau,k)\right] = v^2 - \alpha^2 \sin^2\theta , \eeq again the first integral of~\eqref{eq.pendulum}. \end{proof}

\begin{remark}[The formulas at the poles] \label{rem.chart} The pair $(\theta,\varphi)$ is a chart on the sphere with the poles removed, and the azimuth carries no geometric meaning where $\sin\theta$ vanishes. We interpret the representations~\eqref{eq.rotation} and~\eqref{eq.libration} accordingly, as the continuation through the poles of a curve in the embedded sphere, which the Cartesian components $(\sin\theta\cos\varphi, \sin\theta\sin\varphi, \cos\theta)$ supply without ambiguity; passing a pole sends $\varphi$ to $\varphi + \pi$ and $\sin\theta$ through zero to the opposite sign, and the point on the sphere crosses the pole once. Proposition~\ref{prop.poles} confines the difficulty to the single level set $\ell = 0$, since every other trajectory stays in a band on which the chart is honest. Closure, here and throughout, is the return of the point of the sphere rather than of the coordinate pair, and the verification reported with Figure~\ref{fig.resonance} is carried out in the ambient Cartesian coordinates for that reason. \end{remark}

\begin{figure}[!tbp]
\centering
\includegraphics[width=\linewidth]{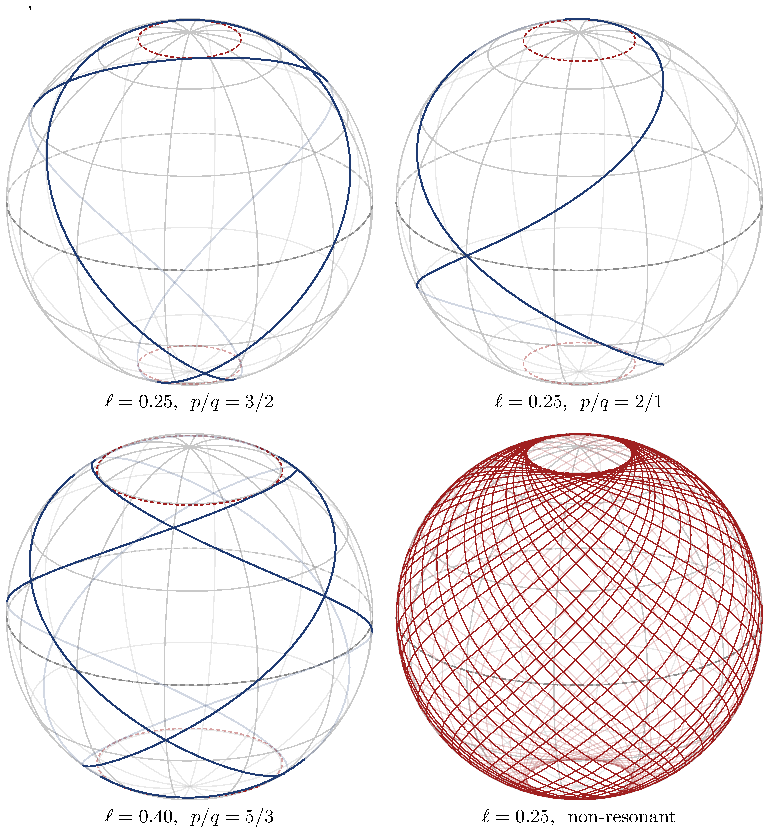}
\caption{Closure on level sets with $\ell \neq 0$, with the speed fixed at $v = 1$ and $\alpha$ tuned until the rotation number of Theorem~\ref{theo.rotation} reaches the value shown. Proposition~\ref{prop.poles} keeps each of these trajectories away from the poles, so every one of them is confined to a band between two parallels of latitude, drawn dashed in red; this is the generic picture, and the pole-crossing motion of Figure~\ref{fig.regimes} belongs to the boundary set $\ell = 0$ alone. The first three panels solve the closure criterion~\eqref{eq.resonance} of Theorem~\ref{theo.closure} at the rotation numbers $p/q = 3/2$, $2/1$ and $5/3$, at $\alpha = 0.394235$, $0.597743$ and $0.416996$ respectively, so that each trajectory returns to its starting point after exactly $q$ polar periods and $p$ turns about the axis. We verified that return independently by integrating the equations of motion, and the two endpoints of each curve agree to about one part in $10^{11}$. The fourth panel sets the rotation number to the golden ratio, an irrational value, so that the trajectory precesses indefinitely and fills its band; forty polar periods are drawn, and the white cap it leaves at either pole is exactly the region Proposition~\ref{prop.poles} excludes. The hemisphere facing away from the reader is drawn faint throughout. } \label{fig.resonance} \end{figure}

The two regimes are dynamically distinct. The ratio $v/\alpha$ separates them. In the rotation regime the polar velocity keeps its sign, for the pendulum first integral~\eqref{eq.pendulum} bounds it below by $\dot\theta^2 \geq v^2 - \alpha^2 > 0$. In the libration regime it vanishes where $\sin\theta = v/\alpha$, which confines the trajectory to a polar cap of that half-angle. The word ``libration'' invites a picture of an oscillation about the equator. The solution~\eqref{eq.libration} supplies a different one, for it carries $\cos\theta = \dn(\alpha\tau,k) \geq \sqrt{1 - k^2} > 0$ for every $\tau$, while $\sin\theta$ vanishes whenever $\sn$ does. Thus the curve stays on one side of the equator, passing through the pole twice in each period. In this sense the two caps carry separate families. At $v = \alpha$ the regimes coincide, the modulus reaches unity, and the elliptic functions degenerate into hyperbolic ones as the polar period diverges. That is the separatrix. We draw all three in Figure~\ref{fig.regimes}.

\begin{figure}[!tbp]
\centering
\includegraphics[width=\linewidth]{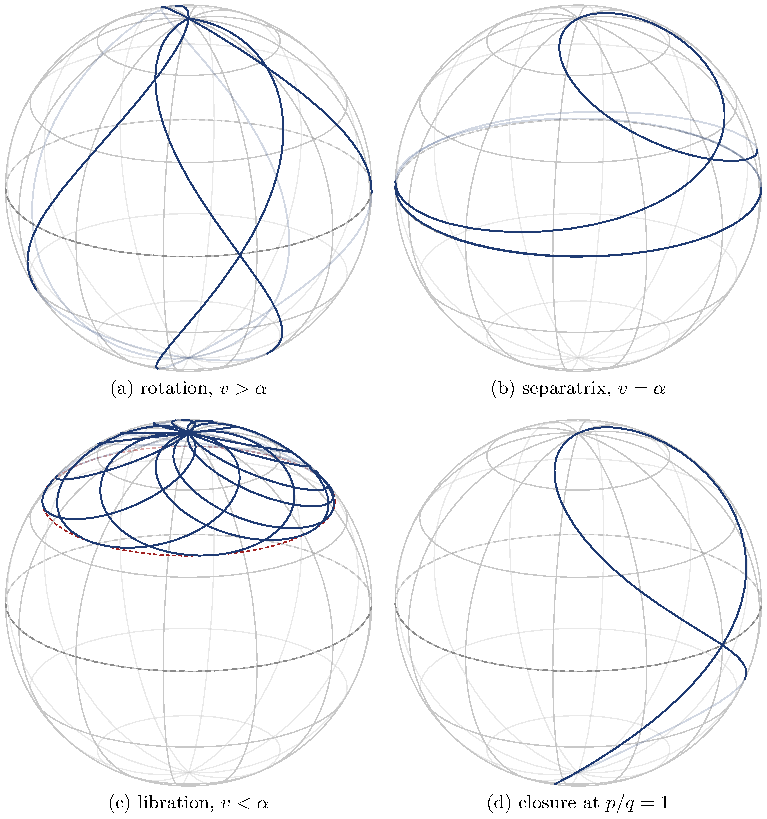}
\caption{The boundary level set $\ell = 0$ of Theorem~\ref{theo.ell0}, with the speed fixed at $v = 1$ and the half-cyclotron frequency $\alpha$ varied between panels. The ratio $v/\alpha$ decides the shape of the motion. Panel (a) takes $\alpha = 0.55$ and lies in the rotation regime, where the curve sweeps through both poles in succession while precessing in azimuth by $4kK(k)$ per polar period. Panel (b) sets $v = \alpha$, where the polar period diverges as the curve winds towards the equator. That value is also the contact threshold of Theorem~\ref{theo.contact}, so panel (b) draws the coincidence Remark~\ref{rem.threshold} records --- the trajectory at which the polar period diverges is the one at which the Liouville field turns tangent to the energy level. Panel (c) takes $\alpha = 1.25$ and lies in the libration regime, confined to the polar cap of half-angle $\arcsin(v/\alpha) = 53.13^\circ$ drawn dashed in red, which it touches on either side of the pole before turning back. Panel (d) solves the resonance $4kK(k) = 2\pi$ of Corollary~\ref{cor.ell0} at $k = 0.792726$, so that $v/\alpha = 1.261470$ and the trajectory closes after one polar period and one turn about the axis (Remark~\ref{rem.resonant}). The equator is dashed in grey throughout, and we draw the hemisphere facing away from the reader faint in every panel.} \label{fig.regimes} \end{figure}

Let us now leave that boundary and treat every trajectory the field carries. Throughout the remainder of this section $\ell \neq 0$, so that Proposition~\ref{prop.poles} confines the motion to a band, and we write $w_- \leq w_+$ for the two roots in $w = x^2$ of the quartic~\eqref{eq.quartic}.

\begin{theo}[The rotation number] \label{theo.rotation} Over one period of the polar angle the azimuth advances by \beq \label{eq.rot} \Delta\varphi = 2\int_{x_-}^{x_+} \left[ \frac{\ell}{1-x^2} + \alpha \right] \frac{\d x}{\sqrt{P(x)}} , \eeq a complete elliptic integral of the third kind. In the variable $w$ the even quartic $P$ becomes the cubic $R(w) = w\,(w_+ - w)(w - w_-)$, and writing $(e_1,e_2,e_3)$ for its roots in decreasing order, so that the motion runs over $[e_2,e_1]$ with $e_1 = w_+$, the polar period and the azimuthal advance are \beq \label{eq.rotlegendre} T_\theta = \frac{2\mu\, K(k)}{\alpha\sqrt{e_1-e_3}} , \qquad \Delta\varphi = \alpha\, T_\theta + \frac{2\mu\, \ell\, \Pi(n,k)}{\alpha\,(1-e_1)\sqrt{e_1-e_3}} , \eeq where $K$ and $\Pi$ are the complete elliptic integrals of the first and third kinds, \beq \label{eq.rotmodulus} k^2 = \frac{e_1-e_2}{e_1-e_3} , \qquad n = -\,\frac{e_1-e_2}{1-e_1} , \eeq and $\mu = 2$ when the band crosses the equator, $\mu = 1$ when it lies within one hemisphere. \end{theo}

\begin{proof} The azimuthal velocity is $\dot\varphi = \ell/\sin^2\theta + \alpha$ by the expression~\eqref{eq.phidot}, with $\sin^2\theta = 1-x^2$, so the advance over one polar period is the integral~\eqref{eq.rot}, the factor of two counting the outward and the return transit between the turning points where $P$ vanishes. The integrand carries simple poles at $x = \pm 1$, which lie outside the band by Proposition~\ref{prop.poles}, so the integral converges and is of the third kind. Moreover, we may reduce it to Legendre form.

For the reduction, we pass to $w$, which gives $\d x = \d w/2\sqrt w$ and $P = \alpha^2 (w_+-w)(w-w_-)$, whence $\d x/\sqrt P = \d w/2\alpha\sqrt{R(w)}$ with $R(w) = w(w_+-w)(w-w_-)$. Thus, a band lying within one hemisphere has $0 < w_-$ and covers $w \in [w_-,w_+]$ once per polar period, so that $(e_1,e_2,e_3) = (w_+,w_-,0)$ and $\mu = 1$. A band crossing the equator has $w_- \leq 0$ and covers $w \in [0,w_+]$ twice, so that $(e_1,e_2,e_3) = (w_+,0,w_-)$ and $\mu = 2$. In either case we substitute $w = e_1 - (e_1-e_2)\sin^2\psi$, which carries $w$ from $e_1$ to $e_2$ as $\psi$ runs over $[0,\pi/2]$ and gives \beq \label{eq.rootsub} \sqrt{R(w)} = (e_1-e_2)\sin\psi\cos\psi \sqrt{e_1-e_3} \sqrt{1-k^2\sin^2\psi} , \eeq with $k$ as in~\eqref{eq.rotmodulus}, while $\d w = -2(e_1-e_2)\sin\psi\cos\psi\ \d\psi$. The factors cancel, and the term of~\eqref{eq.rot} carrying $\alpha$ becomes $2\mu\,\alpha\,K(k)/\alpha\sqrt{e_1-e_3}$, which is $\alpha T_\theta$ with the period of~\eqref{eq.rotlegendre}. For the remaining term, $1-w = (1-e_1)\left[1 + \tfrac{e_1-e_2}{1-e_1}\sin^2\psi\right]$, so the same substitution produces $\Pi(n,k)/(1-e_1)$ with the characteristic $n$ of~\eqref{eq.rotmodulus}, and we assemble the constants into~\eqref{eq.rotlegendre}. \end{proof}

We write $\rho = \Delta\varphi/2\pi$ for the rotation number of a trajectory. By Theorem~\ref{theo.rotation} it depends on $\ell$, $v$ and $\alpha$ alone, so every trajectory of a level set carries the same one. Closure is then a statement about a single number.

\begin{theo}[Closure] \label{theo.closure} A trajectory with $\ell \neq 0$ closes if and only if its rotation number is rational, that is if and only if \beq \label{eq.resonance} \Delta\varphi = 2\pi\, \frac{p}{q} \eeq for coprime integers $p$ and $q$ with $q \geq 1$, in which case it closes after exactly $q$ polar periods and $p$ turns about the polar axis. \end{theo}

\begin{proof} By Proposition~\ref{prop.poles} the motion is confined to a band on which $\theta$ oscillates between the two turning points with period $T_\theta$, so that $\theta(\tau + T_\theta) = \theta(\tau)$ and $\dot\theta(\tau + T_\theta) = \dot\theta(\tau)$ for every $\tau$. The azimuthal velocity depends on the polar angle alone by~\eqref{eq.phidot}. Thus $\varphi(\tau+T_\theta) = \varphi(\tau) + \Delta\varphi$ with $\Delta\varphi$ the constant of Theorem~\ref{theo.rotation}. The trajectory therefore returns to its starting point on the sphere after $q$ polar periods precisely when $q\,\Delta\varphi \in 2\pi\mathbb{Z}$. The least such $q$ is the denominator of $\Delta\varphi/2\pi$ in lowest terms, with $p$ the numerator counting the azimuthal turns. A trajectory whose rotation number is irrational meets each of its points once, and the closure of its image is the whole band. \end{proof}

The criterion~\eqref{eq.resonance} is the principal result of the manuscript. It covers every trajectory the field carries. Indeed, the rotation number depends on the speed only through a pair of dimensionless ratios, and that dependence fixes what the criterion selects.

\begin{prop} \label{prop.scaling} The azimuthal advance depends on the half-cyclotron frequency, the azimuthal integral and the speed through two ratios alone, \beq \label{eq.ratios} \Delta\varphi = \Delta\varphi\left( \frac{\alpha}{v},\ \frac{\ell}{v} \right) , \eeq so that the closure criterion~\eqref{eq.resonance} is a condition on those two numbers. The first of them is the dimensionless combination \beq \label{eq.control} \frac{\alpha}{v} = \frac{qB}{2mv} \eeq of the charge, the mass, the field and the speed, and every quadruple of those four sharing its value carries the same closed trajectory. \end{prop}

\begin{proof} Writing $\alpha = va$ and $\ell = vL$ in the quartic~\eqref{eq.quartic} gives $P = v^2Q$ with \beq \label{eq.scaled} Q(x) = -a^2x^4 + \left[ 2a(L+a) - 1 \right]x^2 + 1 - \left( L+a \right)^2 , \eeq whose coefficients carry no $v$, while the azimuthal velocity~\eqref{eq.phidot} carries one factor of $v$. The two cancel in the advance~\eqref{eq.rot} of Theorem~\ref{theo.rotation}, which leaves~\eqref{eq.ratios}. \end{proof}

In this sense the criterion constrains a ratio. Fixing the speed at unity therefore costs no generality, and every figure below takes $v = 1$.

Which rotation numbers a field can produce is now a question about the range of $\Delta\varphi$. That range is governed by a single value of $\alpha$ that generalises the separatrix of the boundary case. Let us write $y = \sin^2\theta$, in which the polar motion reads \beq \label{eq.ymotion} \dot y^2 = Q(y) = 4(1-y)\left( -\alpha^2 y^2 + E y - \ell^2 \right), \qquad E = v^2 - 2\alpha\ell . \eeq

\begin{prop}[The equatorial value] \label{prop.separatrix} The equator $y = 1$ is a root of $Q$ for every $\alpha$, and it is a double root exactly when \beq \label{eq.sep} \left( \ell + \alpha \right)^2 = v^2 , \eeq whose single admissible root for $0 < \vert \ell \vert < v$ and $\alpha > 0$ is the equatorial value $\alpha = v - \ell$. There the second root of the inner quadratic of~\eqref{eq.ymotion} sits at $y_\ast = \ell^2/(v-\ell)^2$, and its position relative to the equator decides what that value carries. For $\ell < v/2$ the root $y_\ast$ lies below the equator, the trajectory tends to the equator in infinite time and the polar period diverges, which is the separatrix. For $v/2 \leq \ell < v$ the root $y_\ast$ lies at the equator or above it, and the equatorial great circle is then the entire motion of that level set, traversed with finite period. Setting $\ell = 0$ in~\eqref{eq.sep} returns the condition $v = \alpha$ of Theorem~\ref{theo.ell0}. \end{prop}

\begin{proof} The factor $1-y$ in~\eqref{eq.ymotion} vanishes at the equator, which gives the first assertion, and the remaining factor takes the value $-\alpha^2 + E - \ell^2 = v^2 - (\ell+\alpha)^2$ there, which vanishes precisely under the condition~\eqref{eq.sep}. Solving that condition for positive $\alpha$, we obtain $\alpha = v - \ell$, since $\alpha = -v-\ell$ is positive only for $\ell < -v$. The two roots of the inner quadratic multiply to $\ell^2/\alpha^2$ by~\eqref{eq.ymotion}, so the second of them sits at $y_\ast = \ell^2/(v-\ell)^2$, which exceeds unity exactly when $\vert \ell \vert > \vert v - \ell \vert$, that is when $\ell > v/2$. Displaying both roots, the polar motion at the equatorial value reads \beq \label{eq.sepmotion} \dot y^2 = 4\alpha^2 \left( 1 - y \right)^2 \left( y - y_\ast \right) . \eeq For $y_\ast < 1$ the right-hand side of~\eqref{eq.sepmotion} is positive on the interval $(y_\ast,1)$ and carries a double zero at the equator, so the motion runs over that interval and reaches the equator only asymptotically, the period growing logarithmically. For $y_\ast \geq 1$ it is negative at every $y < 1$, so the equator carries the entire motion. \end{proof}

The equatorial value therefore divides the admissible $\alpha$ into a branch below it, on which the band crosses the equator, and a branch above it, on which the band lies within one hemisphere. Let us take the first of them here.

\begin{theo} \label{theo.realise} Fix $v > 0$ and $\ell$ with $0 < \vert\ell\vert < v$, and let $\alpha$ run over the interval $(0, v-\ell)$ on which the band crosses the equator. The rotation number $\rho = \Delta\varphi/2\pi$ is continuous there and tends to the sign of $\ell$ as $\alpha \to 0^+$. As $\alpha \to (v-\ell)^-$ it tends to $+\infty$ when $\ell \leq v/2$, and to the winding $\sqrt{v/(2\ell-v)}$ of the equatorial circle when $v/2 < \ell < v$. Consequently every rational number in \beq \label{eq.range} \mathcal{I}_\ell = \begin{cases} \ \left( \operatorname{sign}(\ell),\ \infty \right) , & \ell \leq v/2 , \\[6pt] \ \left( 1,\ \sqrt{\dfrac{v}{2\ell-v}}\ \right) , & v/2 < \ell < v , \end{cases} \eeq occurs as the rotation number of some $\alpha$ in $(0,v-\ell)$, and by Theorem~\ref{theo.closure} the corresponding trajectory closes after $q$ polar periods and $p$ turns about the polar axis. \end{theo}

\begin{proof} Continuity is that of the roots of $Q$, which are simple on the open interval by Proposition~\ref{prop.separatrix}, together with the continuity of the complete integrals $K$ and $\Pi$ in the modulus and characteristic of~\eqref{eq.rotmodulus}.

For the limit as the field switches off, set $\alpha = 0$ in the quartic~\eqref{eq.quartic}, which degenerates to $P(x) = v^2 - \ell^2 - v^2x^2$, and in the rotation number~\eqref{eq.rot}. Substituting $x = c\sin\psi$ with $c = \sqrt{1-\ell^2/v^2}$ carries $\sqrt{P}$ to $vc\cos\psi$ and gives \beq \label{eq.geodesic} \Delta\varphi\big\vert_{\alpha = 0} = 2\int_{-\pi/2}^{\pi/2} \frac{\ell}{v}\,\frac{\d\psi}{1 - c^2\sin^2\psi} = \frac{2\pi\ell}{v\sqrt{1-c^2}} = 2\pi \operatorname{sign}(\ell) , \eeq since $\sqrt{1-c^2} = \vert\ell\vert/v$. The trajectory is then a great circle, which closes after one turn, as it must. The integrand of~\eqref{eq.rot} and its limits of integration depend continuously on $\alpha$ at $\alpha = 0$, so $\rho \to \operatorname{sign}(\ell)$.

For the limit at the equatorial value, let $\ell \leq v/2$ and split the azimuthal velocity as \beq \label{eq.dphisplit} \frac{\ell}{y} + \alpha = \left( \ell + \alpha \right) + \ell\,\frac{1-y}{y} , \qquad \text{so that} \qquad \Delta\varphi = \left( \ell + \alpha \right) T_\theta + \ell \oint \frac{1-y}{y}\ \d\tau . \eeq The product of the roots of the inner quadratic of~\eqref{eq.ymotion} is $y_-y_+ = \ell^2/\alpha^2$, so $y_- \to \ell^2/(v-\ell)^2 > 0$ as $y_+ \to 1$, and the band stays away from the poles uniformly. Near the upper endpoint $\d\tau$ behaves as $\d y /2\alpha(1-y)\sqrt{y-y_-}$ once $y_+$ has reached $1$, while the factor $(1-y)/y$ vanishes there to first order, so the second integral in~\eqref{eq.dphisplit} remains bounded as that value is approached. Since $\ell + \alpha \to v > 0$, it therefore suffices that $T_\theta$ diverge. Writing that quadratic as $\alpha^2(y_+-y)(y-y_-)$, the upper turning point $y_+$ tends to $1$ there by Proposition~\ref{prop.separatrix}, which places $y_\ast$ at the equator or below it, so the integrand of \beq \label{eq.Tdiverge} T_\theta = 2\mu \int_{y_-}^{\min(y_+,1)} \frac{\d y}{2\alpha\sqrt{(1-y)(y_+-y)(y-y_-)}} \eeq behaves as $\left[(1-y)(y_+-y)\right]^{-1/2}$ near the upper endpoint, which is integrable while $y_+ < 1$ and tends to $(1-y)^{-1}$ as $y_+ \to 1$. The integral therefore grows without bound, logarithmically in the distance to that value when $\ell < v/2$, and faster when $\ell = v/2$, where the lower turning point tends to the equator along with the upper one. With the splitting~\eqref{eq.dphisplit} this gives $\Delta\varphi \to +\infty$.

For $v/2 < \ell < v$ the band collapses onto the equator instead. Near $x = 0$ the quartic~\eqref{eq.quartic} reads $P = \Sigma + \beta x^2 + O(x^4)$, where $\Sigma = v^2 - (\ell+\alpha)^2$ is its constant term and $\beta = 2\alpha(\ell+\alpha) - v^2$ its quadratic coefficient. Letting $\alpha \to (v-\ell)^-$ sends $\Sigma$ to zero from above while $\beta$ tends to $v(v-2\ell) < 0$. The upper turning point $w_+ = \Sigma/\vert \beta \vert + O(\Sigma^2)$ tends to zero with it, so on the band the polar motion is the harmonic oscillation $\dot x^2 = \Sigma - \vert \beta \vert x^2$ up to terms of order $x^4$, of period $2\pi/\sqrt{\vert \beta \vert}$, while the azimuthal velocity $\ell/(1-x^2) + \alpha$ tends uniformly on the band to $\ell + \alpha = v$. Multiplying the two limits gives \beq \label{eq.eqcircle} \Delta\varphi \longrightarrow \frac{2\pi v}{\sqrt{v\left( 2\ell - v \right)}} = 2\pi \sqrt{\frac{v}{2\ell-v}} , \eeq the winding of the equatorial circle that the band collapses onto.

Thus the interval~\eqref{eq.range} is the image of a connected set under a continuous map with those two limits, so it contains every value between them, and in particular every rational one. \end{proof}

Theorem~\ref{theo.realise} produces a value of $\alpha$ for each rational winding, and says nothing about how many. \emph{How many values of $\alpha$ carry a given winding?} That is a question about the derivative of the rotation number, which we now compute in closed form. We write \beq \label{eq.DENS} D = v^2 - 4\alpha\ell , \qquad E = v^2 - 2\alpha\ell , \qquad N = v^2 - 2\ell(\ell+\alpha) , \qquad \Sigma = v^2 - (\ell+\alpha)^2 , \eeq of which $\Sigma$ is positive exactly on the branch below the equatorial value, by Proposition~\ref{prop.separatrix}, and $D$ is positive wherever the turning points are distinct.

\begin{theo} \label{theo.twist} Let $\ell \neq 0$ and let $\alpha$ be an admissible value at which the turning points are distinct. Then \beq \label{eq.twist} \frac{\partial \Delta\varphi}{\partial \alpha} = \frac{1}{D}\left[ E\, T_\theta + \frac{\alpha^2 N}{\Sigma}\, M_2 \right] , \qquad M_2 = \oint \frac{x^2\ \d x}{\sqrt{P(x)}} , \eeq the integrals being taken over one polar period as in~\eqref{eq.rot}. \end{theo}

\begin{proof} Write $y = 1-x^2$, so that the quartic~\eqref{eq.quartic} reads $P = yv^2 - (\ell+\alpha y)^2$ and the integrand of~\eqref{eq.rot} is $f = (\ell+\alpha y)/y$. Differentiating gives $\partial_\alpha P = -2y(\ell+\alpha y)$, whence $f\,\partial_\alpha P = -2(\ell+\alpha y)^2$, and therefore \beq \label{eq.predw} \frac{\partial}{\partial\alpha} \left[ f\, P^{-1/2} \right] = P^{-1/2} + \left(\ell+\alpha y\right)^2 P^{-3/2} = v^2\, y\, P^{-3/2} , \eeq the last step using $(\ell+\alpha y)^2 = yv^2 - P$ once more. We may take the derivative under the integral sign on a contour encircling the cut between the turning points, on which the endpoints do not move, so \beq \label{eq.dwcontour} \frac{\partial \Delta\varphi}{\partial \alpha} = v^2 \oint \left( 1-x^2 \right) P^{-3/2}\ \d x . \eeq

Let us now reduce the exponent. Matching coefficients in $x$ produces polynomials $\mathcal{A} = a_0 + a_2x^2$ and $\mathcal{B} = b_1x + b_3x^3$ with $1 - x^2 = \mathcal{A} P + \tfrac{1}{2}\mathcal{B} P'$, namely \beq \label{eq.pfcoeffs} a_0 = \frac{1}{\Sigma} , \qquad b_3 = \frac{\alpha^2 N}{v^2 D \Sigma} , \qquad a_2 = -2b_3 , \qquad b_1 = \frac{2\alpha\ell(\alpha+\ell)^2 - v^2(\alpha^2+\ell^2)} {v^2 D \Sigma} . \eeq Since $\oint \tfrac{\d}{\d x}\left[ \mathcal{B} P^{-1/2} \right] \d x = 0$ around a closed contour, we have $\oint \tfrac{1}{2}\mathcal{B} P' P^{-3/2}\d x = \oint \mathcal{B}' P^{-1/2}\d x$, and hence \beq \label{eq.pfreduce} \oint \left( 1-x^2 \right) P^{-3/2}\ \d x = \oint \left( \mathcal{A} + \mathcal{B}' \right) P^{-1/2}\ \d x , \qquad \mathcal{A} + \mathcal{B}' = \frac{1}{v^2 D}\left[ E + \frac{\alpha^2 N}{\Sigma}x^2 \right] , \eeq where the expression for $\mathcal{A}+\mathcal{B}'$ follows from~\eqref{eq.pfcoeffs} and $a_2 = -2b_3$. Substituting~\eqref{eq.pfreduce} into~\eqref{eq.dwcontour} and recognising $\oint P^{-1/2}\d x = T_\theta$ gives the identity~\eqref{eq.twist}. \end{proof}

Every ingredient of~\eqref{eq.twist} is positive on this branch except possibly $N$, so the sign of the derivative turns on that one quantity. When $N$ is negative the two terms compete, and the contest is decided by how the weight $P^{-1/2}\d x$ distributes $x^2$ across the band.

\begin{lemma} \label{lem.moment} On the branch below the equatorial value, \beq \label{eq.chebyshev} M_2 \ \leq\ \frac{w_+}{2}\, T_\theta , \eeq where $w_+$ is the larger root in $w = x^2$ of the quartic~\eqref{eq.quartic}. \end{lemma}

\begin{proof} There $w_- \leq 0 < w_+$ and the band is $x^2 \leq w_+$, so we may write $P = \alpha^2(w_+-x^2)(x^2-w_-)$ and substitute $x = \sqrt{w_+}\,\sin\psi$ with $\psi \in [-\pi/2,\pi/2]$. This carries $\d x/\sqrt P$ to $g(\psi)\,\d\psi$ with $g = \alpha^{-1}\left( w_+\sin^2\psi - w_- \right)^{-1/2}$ and $x^2$ to $w_+\sin^2\psi$, so that $M_2/T_\theta = w_+\,\langle g \sin^2\psi\rangle / \langle g \rangle$, the brackets denoting means over $\psi$. Now $g$ is a decreasing function of $\sin^2\psi$, so the two are oppositely ordered. Chebyshev's integral inequality then gives $\langle g\sin^2\psi\rangle \leq \langle g \rangle \langle \sin^2\psi\rangle$. Since $\langle \sin^2\psi\rangle = 1/2$, the bound~\eqref{eq.chebyshev} follows. \end{proof}

\begin{theo} \label{theo.mono} Let $0 < \vert\ell\vert < v$. Then $\Delta\varphi$ is a strictly increasing function of $\alpha$ on the whole branch $0 < \alpha < v-\ell$. \end{theo}

\begin{proof} Throughout, $\Sigma > 0$ below the equatorial value, by Proposition~\ref{prop.separatrix}. The turning points are distinct, since $(\ell-\alpha)^2 \geq 0$ gives $4\alpha\ell \leq (\ell+\alpha)^2 < v^2$ when $\alpha\ell > 0$, while $D \geq v^2$ otherwise; hence $D > 0$, and $E > D > 0$ likewise. Moreover, the period is positive and $M_2 \geq 0$. If $N \geq 0$ every term of~\eqref{eq.twist} is therefore positive and there is nothing to prove.

Suppose $N < 0$. Then $\alpha^2 N/\Sigma < 0$, so Lemma~\ref{lem.moment} bounds the second term from below and \beq \label{eq.reduced} E\,T_\theta + \frac{\alpha^2 N}{\Sigma}M_2 \ \geq\ \frac{T_\theta}{2\Sigma} \left[ 2E\Sigma + \alpha^2 N w_+ \right] . \eeq Indeed, the larger root satisfies $2\alpha^2 w_+ = 2\alpha^2 - E + v\sqrt{D}$, because the discriminant of the quadratic in $w$ is $v^2D$, so that \beq \label{eq.KND} 2\left[ 2E\Sigma + \alpha^2 N w_+ \right] = \mathcal{K} + N v \sqrt{D} , \qquad \mathcal{K} = 4E\Sigma + N\left( 2\alpha^2 - E \right) . \eeq Since $N < 0$, the right-hand side of~\eqref{eq.KND} is positive precisely when $\mathcal{K} > 0$ and $\mathcal{K}^2 > N^2v^2D$. We prove both on the region \beq \label{eq.corner} \mathcal{R} = \left\{ \alpha > 0, \quad \vert\ell\vert < v, \quad \Sigma > 0, \quad N < 0 \right\} , \eeq taking $v = 1$, which costs no generality by Proposition~\ref{prop.scaling}. A direct expansion gives \beq \label{eq.KW} \mathcal{K}^2 - N^2 D = -4\,\Sigma\, \mathcal{W} , \qquad \mathcal{W} = 4\alpha^2\ell^2(\ell+\alpha)^2 - 20\alpha^2\ell^2 - 4\alpha^3\ell + 12\alpha\ell + \alpha^2 - 2 , \eeq so that, $\Sigma$ being positive, the claim is that $\mathcal{K} > 0$ and $\mathcal{W} < 0$ on $\mathcal{R}$.

Indeed, the region $\mathcal{R}$ is bounded, since $\vert\ell\vert < 1$ and $\vert\ell+\alpha\vert < 1$ force $\alpha < 2$. Both $\mathcal{K}$ and $\mathcal{W}$ are polynomials, so each attains its extrema on $\overline{\mathcal{R}}$ either at an interior critical point or on the boundary. Eliminating $\alpha$ from $\nabla \mathcal{W} = 0$ leaves $\ell\,(8\ell^2-1)$, so the critical points of $\mathcal{W}$ are $(\alpha,\ell) = (0,0)$ and $\pm\left(3\sqrt2/4,\ \sqrt2/4\right)$; eliminating $\alpha$ from $\nabla \mathcal{K} = 0$ leaves $\ell\,(2\ell^2-1)^3$, so the critical points of $\mathcal{K}$ are $(0,0)$ and $\pm\left(1/\sqrt2,\ 1/\sqrt2\right)$. At the origin $N = 1 > 0$; at each of the other four $(\ell+\alpha)^2 = 2$, so $\Sigma < 0$. None of them lies in $\mathcal{R}$.

We evaluate both polynomials on each face of the boundary. On $\Sigma = 0$, where $\ell = 1-\alpha$, \beq \label{eq.faceS} \mathcal{W} = -\left( 3\alpha^2-4\alpha+2 \right) \left( 2\alpha-1 \right)^2 , \qquad \mathcal{K} = \left( 2\alpha-1 \right)^2 , \eeq The quadratic factor has discriminant $-8$, hence is positive, so $\mathcal{W} \leq 0$ and $\mathcal{K} \geq 0$ with equality in both only at $\alpha = 1/2$. On $N = 0$, where $2\alpha\ell = 1-2\ell^2$, \beq \label{eq.faceN} \mathcal{W} = -4\ell^2\left( 4\ell^2-1 \right) , \qquad \mathcal{K} = 2\left( 4\ell^2-1 \right) , \eeq Here $N < 0$ together with $\vert\ell+\alpha\vert<1$ forces $\vert\ell\vert > 1/2$, so $\mathcal{W} < 0$ and $\mathcal{K} > 0$ there. On $\alpha = 0$ we have $\mathcal{W} = -2$ and $\mathcal{K} = 3-2\ell^2 > 1$, since $\ell^2 < 1$. The face $\ell = 1$ meets $\mathcal{R}$ nowhere, because $\Sigma = -\alpha(\alpha+2) < 0$ there. On the face $\ell = -1$ the condition $N < 0$ restricts $\alpha$ to $(0,1/2)$, where $\mathcal{W} = 4\alpha^4-4\alpha^3-15\alpha^2-12\alpha-2$ has maximum $-2$ and $\mathcal{K} = -4\alpha^3+6\alpha^2+8\alpha+1$ has minimum $1$.

Hence $\mathcal{W} \leq 0$ and $\mathcal{K} \geq 0$ on $\overline{\mathcal{R}}$, with equality in either only at $(\alpha,\ell) = (1/2,1/2)$, a corner at which $\Sigma$, $N$ and $D$ vanish together, and therefore outside $\mathcal{R}$. On $\mathcal{R}$ itself the inequalities are strict, so $\mathcal{K}^2 > N^2D$ by~\eqref{eq.KW}, the right-hand side of~\eqref{eq.KND} is positive, and~\eqref{eq.reduced} makes $\partial_\alpha \Delta\varphi > 0$. \end{proof}

\begin{corollary}[Uniqueness of $\alpha$] \label{cor.twist} Let $0 < \vert\ell\vert < v$. Each rational number in the set $\mathcal{I}_\ell$ of~\eqref{eq.range} is the rotation number of exactly one value of $\alpha$ on the branch $0 < \alpha < v-\ell$, and carries exactly one closed trajectory there up to the choice of starting point. \end{corollary}

\begin{proof} By Theorem~\ref{theo.mono} the rotation number is a strictly increasing continuous function of $\alpha$ on that branch, so it attains each value of its range exactly once, and Theorem~\ref{theo.realise} identifies that range as $\mathcal{I}_\ell$. Theorem~\ref{theo.closure} then supplies the closed trajectory. \end{proof}

\begin{remark}[Where the argument is delicate] \label{rem.twist} The competition in~\eqref{eq.twist} is real, and no artefact of the estimate. The quantity $N$ does turn negative on the branch, and since $\vert\ell+\alpha\vert < v$ there this asks $\ell > v/2$ for positive $\ell$ and, because a positive $\alpha$ raises $N$ when $\ell$ is negative, $\ell < -v/\sqrt2$ for negative $\ell$. On that corner the bound~\eqref{eq.chebyshev} carries the proof, and it is sharp enough only just. The margin in $\mathcal{W} < 0$ closes as $(\ell,\alpha) \to (v/2,v/2)$, the corner at which the separatrix meets the circular orbit $v^2 = 4\alpha\ell$ and the band collapses to a point. \end{remark}

Let us mark two features of the range below the equatorial value. The lower endpoint is the great circle, which closes after a single turn, so the field only ever increases the winding there. On a level set with $\ell \leq v/2$ the upper end is unbounded, so every rational winding, however large, is produced by an $\alpha$ below the equatorial value. A level set with $\ell > v/2$ instead carries a largest winding, approached as the band collapses onto the equatorial circle.

Let us now cross the equatorial value. \emph{Does the winding keep growing once the band leaves the equator?} We find that it reverses. A symmetry of the hemisphere band turns it.

\begin{prop} \label{prop.hemisphere} Let $0 < \vert\ell\vert < v$. The $\alpha$ above $v-\ell$ carrying a trajectory of that level set form the interval $(v-\ell,\ v^2/4\ell)$ when $0 < \ell < v/2$, the half line $(v-\ell,\ \infty)$ when $-v < \ell < 0$, and the empty set when $v/2 \leq \ell < v$. On that branch the turning points obey $0 < w_- < w_+ < 1$, and writing $s_1 = w_-+w_+$ and $s_2 = w_-w_+$ for their elementary symmetric functions, \beq \label{eq.wDENS} E = \alpha^2\left( 2 - s_1 \right) , \qquad \Sigma = -\,\alpha^2 s_2 , \qquad N = \alpha^2\left( s_1 - 2s_2 \right) , \eeq so that $E$ and $N$ are positive there while $\Sigma$ is negative. At the right endpoint of the first case the turning points merge and the trajectory is the circle of colatitude $\sin\theta_c = 2\ell/v$. \end{prop}

\begin{proof} The roots in $w = x^2$ of the quartic~\eqref{eq.quartic} add to $s_1 = \beta/\alpha^2$ with $\beta = 2\alpha(\ell+\alpha)-v^2$ and multiply to $s_2 = -\Sigma/\alpha^2$. They are real and distinct exactly when $D > 0$, the discriminant of that quadratic being $v^2D$. Since $\beta = 2\alpha^2 - E$, the first two identities of~\eqref{eq.wDENS} follow at once. The third follows from them because $N = 2\alpha^2 + 2\Sigma - E$, which is the definition~\eqref{eq.DENS} rewritten.

Now, a band lies within one hemisphere exactly when $w_- > 0$, which asks $s_2 > 0$ and $s_1 > 0$, that is $\Sigma < 0$ and $\beta > 0$. The first holds exactly above the equatorial value, by Proposition~\ref{prop.separatrix}. For the second, $\beta$ increases with $\alpha$ on the branch, since $\partial_\alpha \beta = 4\alpha + 2\ell$ while $\alpha > v - \ell$ there. It takes the value $v(v-2\ell)$ at $\alpha = v-\ell$, and the value $v^2(v^2-4\ell^2)/8\ell^2$ at $\alpha = v^2/4\ell$. Both are positive exactly when $\ell < v/2$. Hence for $\ell < v/2$ we have $\beta > 0$ on the whole branch. For negative $\ell$ that branch carries no constraint from $D = v^2 + 4\alpha\vert\ell\vert > 0$; for positive $\ell$ it stops where $D$ vanishes at $\alpha = v^2/4\ell$. For $\ell \geq v/2$ instead $\beta \leq 0$ up to $\alpha = v^2/4\ell$ while $s_2 > 0$, so both turning points are non-positive and no band exists. Beyond that value $D < 0$, and they are complex.

For the upper turning point, $2\alpha^2 w_+ = 2\alpha^2 - E + v\sqrt{D}$ and $E = (v^2+D)/2 > 0$, so $w_+ < 1$ amounts to $v^2 D < E^2$. Indeed, $E^2 - v^2D = 4\alpha^2\ell^2 > 0$ for $\ell \neq 0$. Finally $D$ vanishes at $\alpha = v^2/4\ell$, where $E = v^2/2$ and the merged turning point is $w = 1 - E/2\alpha^2 = 1 - 4\ell^2/v^2$, which is $\cos^2\theta_c$ for the stated colatitude. \end{proof}

Above the equatorial value Proposition~\ref{prop.hemisphere} settles the sign of $N$ at once. The sign of $\Sigma$ turns instead, so the two terms of~\eqref{eq.twist} again compete.

\begin{lemma} \label{lem.involution} Write $\langle f \rangle$ for the mean of $f$ over one polar period against the quadrature~\eqref{eq.rot}, taken in the variable $w = \cos^2\theta$. Above the equatorial value that mean is invariant under the involution $w \mapsto s_2/w$ of the band, so that \beq \label{eq.involution} \left\langle w \right\rangle = s_2\left\langle w^{-1} \right\rangle \qquad\text{and}\qquad \frac{M_2}{T_\theta} = \left\langle w \right\rangle \geq \sqrt{s_2} , \eeq the geometric mean of the two turning points, with equality exactly at the circular orbit. \end{lemma}

\begin{proof} In the variable $w$, $\d x/\sqrt{P}$ becomes $\d w/2\alpha\sqrt{R(w)}$ with $R(w) = w(w_+-w)(w-w_-)$, as in the proof of Theorem~\ref{theo.rotation}, so that $M_2/T_\theta$ is the mean $\langle w \rangle$ against $\d w/\sqrt{R}$ over the band $[w_-,w_+]$. Moreover, the map $\sigma(w) = s_2/w$ carries that band onto itself and exchanges its endpoints. A direct computation gives \beq \label{eq.Rinvolution} R\left( \sigma(w) \right) = \frac{s_2^2}{w^4}\,R(w) , \qquad \text{while} \qquad \vert \d\sigma \vert = \frac{s_2}{w^2}\,\d w , \eeq so that $\vert\d\sigma\vert/\sqrt{R(\sigma)} = \d w/\sqrt{R(w)}$ and the measure is invariant. Hence $\langle w \rangle = \langle \sigma(w) \rangle = s_2\langle w^{-1}\rangle$. The Cauchy--Schwarz inequality applied to $\sqrt w$ and $1/\sqrt w$ gives $\langle w \rangle\langle w^{-1}\rangle \geq 1$. Combining it with that identity leaves $\langle w \rangle^2 \geq s_2$. Equality in Cauchy--Schwarz asks that $w$ be constant along the trajectory, which happens exactly when the turning points merge. \end{proof}

In this sense the mean of $\cos^2\theta$ is at least the geometric mean of its two extreme values.

\begin{theo} \label{theo.mono2} On the branch above the equatorial value $\Delta\varphi$ is a strictly decreasing function of $\alpha$. \end{theo}

\begin{proof} By Proposition~\ref{prop.hemisphere} the turning points obey $0 < w_- < w_+ < 1$, so that $D > 0$, $E > 0$, $N > 0$ and $\Sigma < 0$ there. Substituting the identities~\eqref{eq.wDENS} into the derivative~\eqref{eq.twist} of Theorem~\ref{theo.twist} gives \beq \label{eq.twist2} \frac{\partial \Delta\varphi}{\partial \alpha} = \frac{\alpha^2 T_\theta}{D} \left[ \left( 2 - s_1 \right) - \frac{s_1-2s_2}{s_2}\left\langle w \right\rangle \right] , \eeq so the derivative is negative exactly when $(s_1-2s_2)\langle w \rangle > s_2(2-s_1)$. The coefficient $s_1-2s_2 = w_-(1-w_+) + w_+(1-w_-)$ is positive, so Lemma~\ref{lem.involution} bounds the left-hand side from below by $(s_1-2s_2)\,g$ with $g = \sqrt{s_2}$, and \beq \label{eq.amgm} \left( s_1 - 2s_2 \right) g - s_2\left( 2 - s_1 \right) = g\left( 1+g \right)\left( s_1 - 2g \right) = g\left( 1+g \right)\left( \sqrt{w_+}-\sqrt{w_-} \right)^2 , \eeq which is positive away from the circular orbit by the inequality between the arithmetic and the geometric mean. The bracket of~\eqref{eq.twist2} is therefore negative. \end{proof}

The involution shortens that proof, and it is worth saying why it is available here. A band within one hemisphere keeps $\cos^2\theta$ away from zero, so the reciprocal weight introduced by the passage to $w$ pairs each point of the band with another, and the exact identity~\eqref{eq.involution} follows. A band crossing the equator reaches $w = 0$, which the involution sends to infinity.

\begin{theo} \label{theo.realise2} Let $0 < \vert\ell\vert < v$ and let $\alpha$ run over the branch above the equatorial value. The rotation number decreases continuously from $+\infty$ to $\sec\theta_c = v/\sqrt{v^2-4\ell^2}$ at the circular orbit when $0 < \ell < v/2$, and from $+\infty$ to zero as $\alpha \to \infty$ when $-v < \ell < 0$. Every rational number in \beq \label{eq.range2} \mathcal{J}_\ell = \begin{cases} \ \left( \sec\theta_c ,\ \infty \right) , & 0 < \ell < v/2 , \\[4pt] \ \left( 0 ,\ \infty \right) , & -v < \ell < 0 , \end{cases} \eeq is therefore the rotation number of exactly one value of $\alpha$ on that branch, and carries exactly one closed trajectory there up to the choice of starting point. \end{theo}

\begin{proof} Continuity and the divergence at the equatorial value are those of Theorem~\ref{theo.realise}, the splitting~\eqref{eq.dphisplit} applying verbatim on this side, where the upper turning point in $y$ is $1-w_-$ and again tends to the equator while the lower one $1-w_+$ stays away from the poles.

At the circular orbit let $\alpha \to (v^2/4\ell)^-$, so that both turning points tend to $w_c = 1-4\ell^2/v^2$ by Proposition~\ref{prop.hemisphere}. The polar motion in $w$ reads $\dot w^2 = 4\alpha^2 R(w)$ with $R$ as in Lemma~\ref{lem.involution}, and near the merging turning points $R(w) = w_c(w_+-w)(w-w_-)$ up to terms vanishing with $w_+-w_-$, so the motion is harmonic of frequency $2\alpha\sqrt{w_c}$ and the period tends to $\pi/(\alpha\sqrt{w_c})$, while the azimuthal velocity tends to $\ell/(1-w_c) + \alpha = 2\alpha$. Multiplying the two limits leaves \beq \label{eq.circwinding} \Delta\varphi \longrightarrow \frac{2\pi}{\sqrt{w_c}} = \frac{2\pi v}{\sqrt{v^2-4\ell^2}} = 2\pi\sec\theta_c . \eeq

For $\ell < 0$ and $\alpha \to \infty$, let us write the azimuthal velocity as $\dot\varphi = \alpha\left( 1 - \sqrt{y_-y_+}/y \right)$, using $y_-y_+ = \ell^2/\alpha^2$ from~\eqref{eq.ymotion}. It vanishes at the geometric mean of the turning points, which lies inside the band, so the azimuth reverses within each polar period and $\vert\dot\varphi\vert$ is at most $\alpha(r-1)$ with $r = \sqrt{y_+/y_-}$. Now $r-1 \leq (r^2-1)/2 = (y_+-y_-)/2y_-$, with $y_+-y_- = v\sqrt{D}/\alpha^2$, $y_- = \ell^2/\alpha^2y_+$ and $y_+ \leq E/\alpha^2$, so that $\vert\dot\varphi\vert \leq v\sqrt{D}\,E/(2\ell^2\alpha)$. The period obeys $T_\theta \leq \pi/(\alpha\sqrt{1-y_+})$, which is at most $2\pi/\alpha$ once $y_+ \leq 3/4$, and therefore \beq \label{eq.tailbound} \left\vert \Delta\varphi \right\vert \ \leq\ T_\theta \max \left\vert \dot\varphi \right\vert \ \leq\ \frac{\pi v \sqrt{D}\,E}{\ell^2\alpha^2} \ \longrightarrow\ 0 , \eeq since $\sqrt{D}\,E$ grows as $\alpha^{3/2}$ while the denominator grows as $\alpha^2$.

A continuous and strictly decreasing function attains each value of its range exactly once. Theorem~\ref{theo.mono2} supplies the monotonicity, so each rational number of~\eqref{eq.range2} comes from a single value of $\alpha$. Theorem~\ref{theo.closure} then supplies its closed trajectory. \end{proof}

The two limits of Theorem~\ref{theo.realise2} name the orbits on which the branch ends. The circular orbit winds $\sec\theta_c$ times about the polar axis for each oscillation of its neighbours, so on a level set with $0 < \ell < v/2$ every closed trajectory above the equatorial value winds more than that. A level set with $\ell < 0$ ends instead on a loop that shrinks to a point, and the winding it carries vanishes with it.

\begin{remark}[The large-field tail and the Larmor disc] \label{rem.larmor} The bound~\eqref{eq.tailbound} is generous, and the true rate carries a geometric reading that we record here. On a level set with $\ell < 0$ the substitution $y = (\vert\ell\vert/\alpha)\tau$ sends the band to $[\tau_-,\tau_+]$ with $\tau_- \tau_+ = 1$ and $\tau_-+\tau_+ = 2 + v^2/\vert\ell\vert\alpha$, and the two leading integrals cancel exactly because that product is unity. The first surviving term comes from expanding $(1-y)^{-1/2}$, which leaves $\alpha^2 \Delta\varphi \to \pi v^2/4$, so that \beq \label{eq.larmor} \Delta\varphi \ \sim\ \pi \left( \frac{v}{2\alpha} \right)^2 . \eeq The trajectory is a small loop of geodesic curvature $2\alpha\cos\theta/v$ near a pole, so $v/2\alpha$ is its radius of curvature and the right-hand side of~\eqref{eq.larmor} is the area it encloses. The azimuthal advance over one period is therefore the solid angle of the Larmor disc, which is the holonomy of the round sphere around it. \end{remark}

Collecting the two branches answers the question of Theorem~\ref{theo.realise} in full.

\begin{corollary} \label{cor.count} Fix $v > 0$ and a level set with $0 < \vert\ell\vert < v$, and let $\rho$ be rational. A level set with $0 < \ell < v/2$ carries two values of $\alpha$ at which the rotation number is $\rho$ when $\rho > \sec\theta_c$ and exactly one when $1 < \rho \leq \sec\theta_c$. A level set with $v/2 \leq \ell < v$ carries exactly one when $1 < \rho < \sqrt{v/(2\ell-v)}$. A level set with $-v < \ell < 0$ carries two when $\rho > 0$ and exactly one when $-1 < \rho \leq 0$. Every other rational number is the rotation number of no value of $\alpha$ on that level set. \end{corollary}

\begin{proof} The admissible $\alpha$ split at the equatorial value, by Proposition~\ref{prop.separatrix}. Below it the rotation number increases strictly, by Theorem~\ref{theo.mono}, with range $\mathcal{I}_\ell$ of~\eqref{eq.range}; above it the rotation number decreases strictly, by Theorem~\ref{theo.mono2}, with range $\mathcal{J}_\ell$ of~\eqref{eq.range2}. Each branch therefore carries one value of $\alpha$ for each rational number of its own range and none for any other, so the count is the number of those two sets containing $\rho$. For $0 < \ell < v/2$ they are $(1,\infty)$ and $(\sec\theta_c,\infty)$ with $\sec\theta_c > 1$; for $v/2 \leq \ell < v$ the second branch is empty by Proposition~\ref{prop.hemisphere}; for $-v < \ell < 0$ they are $(-1,\infty)$ and $(0,\infty)$. \end{proof}

Thus a rational winding is carried by at most two values of $\alpha$ on any level set, and the second occurs exactly when the level set admits a band within one hemisphere. The two trajectories are geometrically distinct, one crossing the equator, one confined to a hemisphere, and they close after the same number of polar periods.

Figures~\ref{fig.classI},~\ref{fig.classII} and~\ref{fig.classIII} draw one level set of each class, with $\alpha$ swept across the branches the classification names.

\begin{figure}[!tbp]
\centering
\includegraphics[width=\linewidth]{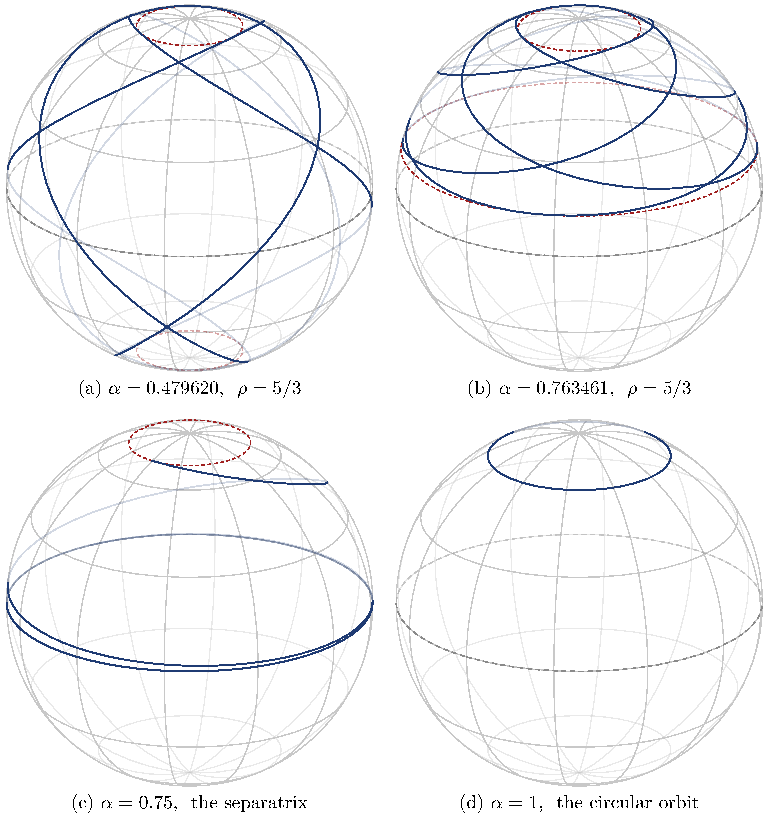}
\caption{A level set with $0 < \ell < v/2$, drawn at $\ell = 0.25$ and $v = 1$, where both branches are present. Panels (a) and (b) carry the same winding $\rho = 5/3$ at two different values of $\alpha$, which makes the count of Corollary~\ref{cor.count} visible. In (a) $\alpha$ sits below the equatorial value $v - \ell = 0.75$ and the band crosses the equator, while in (b) it sits above it and the band lies within the northern hemisphere; both trajectories return to their starting point after three polar periods and five turns about the axis. Panel (c) sets $\alpha$ exactly at the equatorial value, releases the trajectory at the parallel $\cos\theta = 0.9428$, and follows it as it spirals onto the equator, which it approaches asymptotically and reaches at colatitude $89.9993^\circ$ after the time drawn (Proposition~\ref{prop.separatrix}). Panel (d) sets $\alpha = v^2/4\ell = 1$, the right endpoint of the upper branch, where the two turning points merge and the orbit is the circle of colatitude $\theta_c = 30^\circ$, the neighbours of which wind $\sec\theta_c = 1.1547$ times for each oscillation (Theorem~\ref{theo.realise2}). The bounding parallels of each band are dashed in red.} \label{fig.classI} \end{figure}

\begin{figure}[!tbp]
\centering
\includegraphics[width=\linewidth]{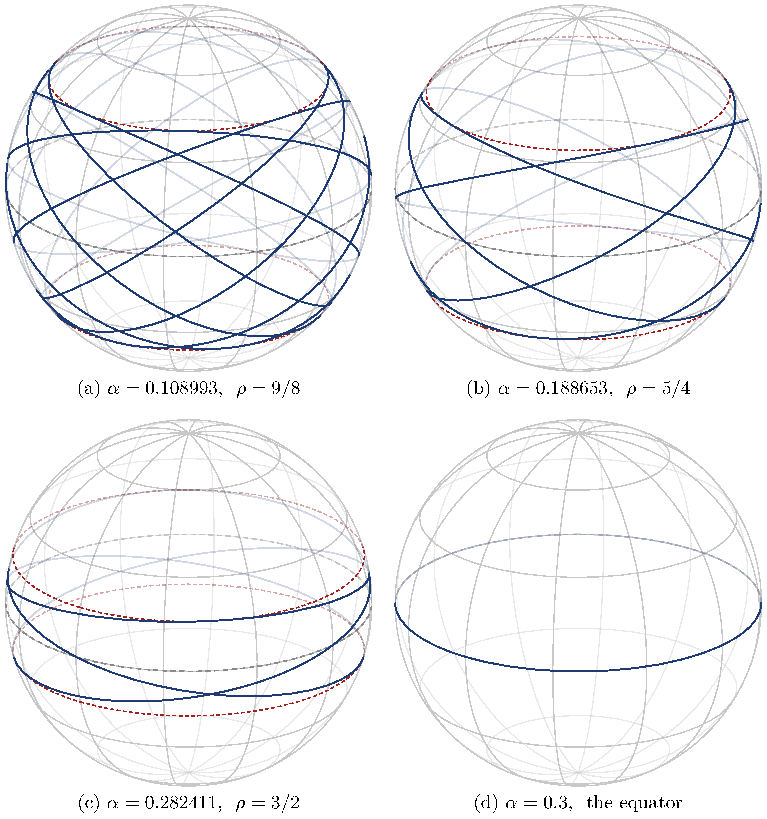}
\caption{A level set with $v/2 \leq \ell < v$, drawn at $\ell = 0.7$ and $v = 1$, where the equatorial value ends the family. Panels (a), (b) and (c) show closed orbits at $\rho = 9/8$, $\rho = 5/4$ and $\rho = 3/2$, and the band tightens about the equator as the winding grows. Panel (d) sets $\alpha = v - \ell = 0.3$, where the band has closed onto the equatorial great circle and that circle is the entire motion of the level set (Proposition~\ref{prop.separatrix}). Its winding $\sqrt{v/(2\ell-v)} = 1.5811$ is the largest the level set admits, so the windings available here fill the interval $(1, 1.5811)$ and stop there (Theorem~\ref{theo.realise}). Every winding this level set admits belongs to a single $\alpha$, since the branch above the equatorial value is empty (Proposition~\ref{prop.hemisphere}).} \label{fig.classII} \end{figure}

\begin{figure}[!tbp]
\centering
\includegraphics[width=\linewidth]{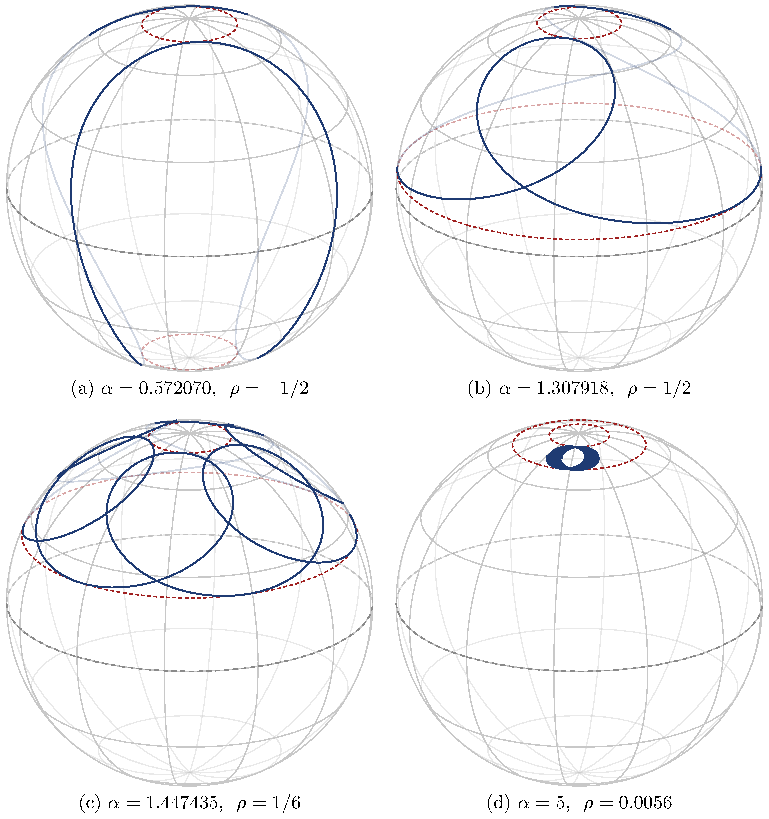}
\caption{A level set with $-v < \ell < 0$, drawn at $\ell = -0.3$ and $v = 1$, where the branch above the equatorial value runs to arbitrarily large $\alpha$. Panel (a) sits below the equatorial value $v - \ell = 1.3$ at the winding $\rho = -1/2$, so the azimuth runs backwards; this level set admits it because the rotation number tends to $-1$ as the field switches off (Theorem~\ref{theo.realise}), the great circle being traversed in the retrograde sense. Panels (b) and (c) sit above the equatorial value at $\rho = 1/2$ and $\rho = 1/6$, where the band lies within the northern hemisphere and withdraws toward the pole as $\alpha$ grows. Panel (d) sets $\alpha = 5$ and draws nine polar periods, by which point the trajectory is a small loop of geodesic curvature close to $2\alpha/v$ that leaves the pole outside itself and precesses through $0.0056$ of a turn in each period. The winding falls to zero along that branch at the rate $\Delta\varphi \sim \pi(v/2\alpha)^2$ of Remark~\ref{rem.larmor}, which is the area the loop encloses.} \label{fig.classIII} \end{figure}

Let us now take the criterion to the boundary level set, where Theorem~\ref{theo.ell0} has already supplied the motion in closed form.

\begin{corollary} \label{cor.ell0} On the level set $\ell = 0$ the azimuthal advance over one polar period is $\Delta\varphi = \alpha T_\theta$ with $T_\theta = 4K(k)/\max(\alpha,v)$, so that the closure criterion~\eqref{eq.resonance} reads \beq \label{eq.resonance0} 4\,k\,K(k) = 2\pi\,\frac{p}{q} \qquad \text{in the rotation regime} \quad (v > \alpha), \eeq \beq \label{eq.resonancelib} 4\,K(k) = 2\pi\,\frac{p}{q} \qquad \text{in the libration regime} \quad (v < \alpha). \eeq The map $k \mapsto 4kK(k)$ is strictly increasing on $(0,1)$ and carries it onto $(0,\infty)$, so every positive rational rotation number is attained by exactly one value of $k$, and hence by exactly one value of $\alpha$. The map $k \mapsto 4K(k)$ carries $(0,1)$ onto $(2\pi,\infty)$, so a trapped trajectory always winds more than once about the polar axis in one polar oscillation. \end{corollary}

\begin{proof} Setting $\ell = 0$ in the rotation number~\eqref{eq.rot} leaves $\Delta\varphi = \alpha T_\theta$, since the term carrying the poles at $x = \pm 1$ vanishes with $\ell$. In the rotation regime the polar angle increases monotonically and \beq \label{eq.Trot} T_\theta = \int_0^{2\pi} \frac{\d\theta}{\sqrt{v^2 - \alpha^2\sin^2\theta}} = \frac{4}{v}\int_0^{\pi/2} \frac{\d\theta}{\sqrt{1-k^2\sin^2\theta}} = \frac{4K(k)}{v}, \eeq with $k = \alpha/v$, by the symmetry of the integrand about $\theta = \pi/2$ and its periodicity. In the libration regime the polar angle turns where $\sin\theta = k$, with $k = v/\alpha$, and the period is four times the quarter oscillation, \beq \label{eq.Tlib} T_\theta = 4\int_0^{\arcsin k} \frac{\d\theta}{\sqrt{v^2-\alpha^2\sin^2\theta}} = \frac{4}{\alpha}\int_0^{\pi/2} \frac{\d \chi}{\sqrt{1-k^2\sin^2 \chi}} = \frac{4K(k)}{\alpha}, \eeq where we use the substitution $\sin\theta = k \sin\chi$ in the second equality. Multiplying by $\alpha$ gives $4kK(k)$ and $4K(k)$ respectively. Theorem~\ref{theo.closure} then supplies the criterion. For the monotonicity, $K$ is strictly increasing on $(0,1)$ and so is the identity, so their product is. As $k \to 0^+$ we have $K(k) \to \pi/2$ and $4kK(k) \to 0$, while as $k \to 1^-$ we have $K(k) \to \infty$ and $4kK(k) \to \infty$; the map is continuous, so it is onto $(0,\infty)$. The same limits give $4K(k) \to 2\pi$ and $4K(k) \to \infty$ for the libration map, whose range is therefore $(2\pi,\infty)$. \end{proof}

The boundary case is a limit of the general one. That limit carries a jump of exactly one turn. Letting $\ell \to 0$ in the rotation regime sends $e_1$ to unity and the characteristic $n$ of~\eqref{eq.rotmodulus} to $-\infty$, so that $\Pi(n,k)$ vanishes like $\tfrac{1}{2}\pi\left(-n\right)^{-1/2}$, while $1-e_1$ vanishes like $\ell^2/v^2$ by the value~\eqref{eq.Ppole} of the quartic at the pole. The two rates conspire. The term carrying $\Pi$ in~\eqref{eq.rotlegendre} therefore tends to $2\pi\operatorname{sign}(\ell)$ rather than to zero, so that \beq \label{eq.boundaryjump} \Delta\varphi \longrightarrow 4kK(k) + 2\pi \operatorname{sign}(\ell) \qquad \text{as} \quad \ell \to 0 . \eeq That extra turn is geometry rather than bookkeeping. Proposition~\ref{prop.poles} keeps a trajectory of small azimuthal integral off the poles, so it rounds each of them in the sense of $\ell$. The two passages of a polar period contribute a full turn between them, which the pole-crossing trajectory of Theorem~\ref{theo.ell0} never makes. The closure criterion~\eqref{eq.resonance} is untouched by the jump, a rotation number changing by an integer.

Taken together, Theorems~\ref{theo.closure} and~\ref{theo.realise} are a selection rule. The closed orbits occupy a set of $\alpha$ that is dense in the admissible interval and of Lebesgue measure zero within it, and Corollary~\ref{cor.count} counts those that carry a given winding, while Corollary~\ref{cor.ell0} does the same on the boundary set $\ell = 0$, where the two branches are the libration and the rotation regime. Figure~\ref{fig.rho} draws the rotation number against $\alpha$ on three level sets, with both branches of each, so the whole selection rule fits in a single picture. Moreover, the analogy with the spacing of Landau levels is one of structure alone --- the discreteness here belongs to the set of fields admitting a closed classical orbit, and we perform no quantisation.

\begin{figure}[!tbp]
\centering
\includegraphics[width=\linewidth]{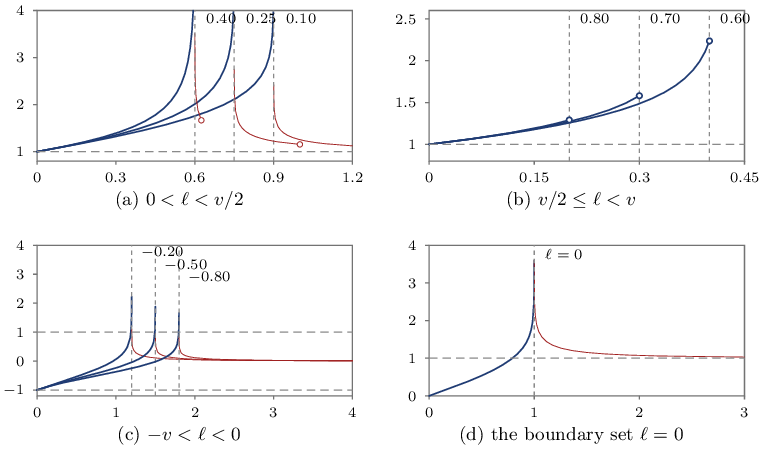}
\caption{The rotation number as a function of $\alpha$, one panel for each class of level set, with the speed fixed at $v = 1$ and the numbers along the top naming the level sets by their value of $\ell$. In every panel the fine vertical dash marks the equatorial value $\alpha = v-\ell$, the branch below it is drawn in blue and the branch above it in red, and the long horizontal dash marks $\rho = \pm 1$, the winding of a great circle. Panel (a) carries both branches, the blue one rising from $1$ without bound and the red one falling to the open circle at the circular orbit, whose winding $\sec\theta_c$ is the least the hemisphere band attains. In panel (b) the branch above the equatorial value is empty by Proposition~\ref{prop.hemisphere}, so the blue branch is the whole family and it ends at the open circle on the equatorial circle, of winding $\sqrt{v/(2\ell-v)}$, the largest such a level set admits. Panel (c) starts at $\rho = -1$, the retrograde great circle, and its red branch runs to arbitrarily large $\alpha$ with the winding decaying to zero at the rate $\pi(v/2\alpha)^2$ of Remark~\ref{rem.larmor}; its three level sets carry separatrices at $\alpha = 1.2$, $1.5$ and $1.8$, in that order from the left. Panel (d) is the boundary set, where Corollary~\ref{cor.ell0} gives $\rho$ in closed form as $4kK(k)/2\pi$ on the rotation branch and $4K(k)/2\pi$ on the libration branch, so that the same two-branch shape appears with the rotation and libration regimes in the roles the two bands play elsewhere. Each blue branch increases strictly by Theorem~\ref{theo.mono} and each red branch decreases strictly by Theorem~\ref{theo.mono2}, so each meets every rational number of its own range exactly once, which is the count of Corollary~\ref{cor.count}.} \label{fig.rho} \end{figure}

\begin{remark}[The simplest resonance on the boundary set] \label{rem.resonant} On the level set $\ell = 0$ the simplest resonance is $p/q = 1$, where the trajectory closes after a single polar period and a single turn about the axis. Solving $4kK(k) = 2\pi$ gives \beq \label{eq.k1} k = 0.792726\ldots , \qquad \text{that is} \qquad \frac{v}{\alpha} = 1.261470\ldots , \eeq so the closure occurs at the field strength $B = 2mvk/q$ with that value of $k$. \end{remark}

\begin{remark}[Charge asymmetry] \label{rem.charge} The reduction depends on the field through $\alpha^2$, so the regimes and the resonance conditions~\eqref{eq.resonance} and~\eqref{eq.resonancelib} are unchanged under $q \to -q$. The azimuthal drift~\eqref{eq.phifrozen} reverses. The closed orbits of the two charges at a common speed are therefore mirror images in the azimuth rather than distinct families, which is a sharper statement than the irreversibility noted in general in Section~\ref{sec.problem}, and it holds only on the level set $\ell = 0$. On the level sets $\ell \neq 0$ the quartic~\eqref{eq.quartic} is not invariant under $\alpha \to -\alpha$ at fixed $\ell$, and the two charges do have distinct families. \end{remark}

%%============================================================
\section{Exactness and the contact threshold} \label{sec.contact}
%%============================================================

\emph{What does the vanishing of the flux buy us in phase space?} Let us take up that question, for the answer separates this field from the monopole one at a stroke.

\begin{prop}[Exactness] \label{prop.exact} The restricted uniform field~\eqref{eq.field} has vanishing total flux through the sphere, its cohomology class in $H^2(\Sph^2;\Reals)$ is therefore zero, and it is globally exact, \beq \label{eq.potential} F = \d A , \qquad A = \frac{B}{2}\sin^2\theta\ \d\varphi , \eeq where $A$ is a smooth $1$-form on the whole of $\Sph^2$, including both poles, whose pointwise magnitude is \beq \label{eq.Anorm} \left\vert A \right\vert_g = \frac{B}{2}\, \sin\theta \ \leq\ \frac{B}{2} , \eeq with equality on the equator alone. \end{prop}

\begin{proof} Indeed, the flux is \beq \label{eq.flux} \int_{\Sph^2} F = \int_0^{2\pi}\!\!\int_0^\pi B \sin\theta\cos\theta\ \d\theta\ \d\varphi = 2\pi B \int_0^\pi \sin\theta\cos\theta\ \d\theta = 0, \eeq the last integrand being antisymmetric about the equator. Since $H^2(\Sph^2; \Reals) \cong \Reals$ is detected by the flux, the class of $F$ vanishes and $F$ is exact. We verify that the form $A$ of~\eqref{eq.potential} is a potential by the computation $\d A = B\sin\theta\cos\theta\ \d\theta\wedge\d\varphi = F$. Moreover, smoothness at the poles is the assertion that $\sin^2\theta\ \d\varphi$ extends across $\theta = 0$ and $\theta = \pi$, which holds because $\sin^2\theta$ vanishes there to second order while $\d\varphi$ blows up at first order in the distance to the pole. The magnitude~\eqref{eq.Anorm} follows from the inverse metric, since $\vert A\vert^2_g = g^{\varphi\varphi}A_\varphi^2 = \sin^{-2}\theta \cdot B^2\sin^4\theta/4$. \end{proof}

It is worth setting the monopole beside this field, for the two differ exactly here. There the field is closed and carries flux, the magnetic charge, and the obstruction to a global potential is the whole content of the Dirac quantisation~\cite{lopez2025noether}. Here the flux vanishes because the ambient field is uniform and the sphere closes upon itself --- the northern and southern hemispheres see opposite fluxes. \emph{In this sense we trade a topological obstruction for a threshold.}

\begin{prop}[The Ma\~n\'e critical value] \label{prop.mane} The exact magnetic flow of the restricted uniform field~\eqref{eq.field} has Ma\~n\'e strict critical value \beq \label{eq.mane} c_0 = \inf_{\nu \in C^\infty(\Sph^2)}\ \sup_{x \in \Sph^2}\ \frac{1}{2m}\left\vert \d_x \nu - qA \right\vert^2_g = \frac{1}{2}m\alpha^2 , \eeq where $\nu$ ranges over the smooth functions on the sphere, and the infimum is attained at the constant ones. \end{prop}

\begin{proof} Taking $\nu$ constant leaves $\vert qA \vert^2_g /2m$, which is $q^2B^2\sin^2\theta /8m$ by the magnitude~\eqref{eq.Anorm}, and its supremum over the sphere is $q^2B^2/8m$, attained on the equator. That number is $\tfrac{1}{2}m\alpha^2$ by the abbreviation~\eqref{eq.alpha}, so the infimum is at most $\tfrac{1}{2}m\alpha^2$.

For the reverse inequality, we note first that the rotations about the polar axis preserve both $g$ and $A$, and that $p \mapsto \vert p - qA\vert^2_g/2m$ is convex on each fibre. Averaging a given $\nu$ over that circle action therefore produces an axially symmetric $\nu$ whose supremum is no larger, by Jensen's inequality and the invariance, so we may seek the infimum among functions of $\theta$ alone. For such a $\nu$ the differential is $\nu'(\theta)\,\d\theta$, which is $g$-orthogonal to $A$, and hence \beq \label{eq.manebound} \frac{1}{2m}\left\vert \d \nu - qA \right\vert^2_g = \frac{1}{2m}\left[ \nu'(\theta)^2 + \frac{q^2B^2}{4}\sin^2\theta \right] \ \geq\ \frac{q^2B^2}{8m}\sin^2\theta . \eeq Evaluating the bound~\eqref{eq.manebound} on the equator gives $\sup_x \geq q^2B^2/8m = \tfrac{1}{2}m\alpha^2$ for every such $\nu$, which is the reverse inequality. \end{proof}

Thus the threshold we compute is the one the general theory names. With the value~\eqref{eq.mane} in hand, the criterion recalled in Section~\ref{sec.intro}~\cite{contreras2004periodic} becomes explicit for this field, and we exhibit the Liouville vector field that realises it.

\begin{theo}[The contact threshold] \label{theo.contact} The twisted symplectic form~\eqref{eq.twisted} is exact on $T^*\Sph^2$, \beq \label{eq.twistedexact} \omega_F = -\d\lambda' , \qquad \lambda' = \lcan - q\, \pi^*A , \eeq and the vector field \beq \label{eq.liouville} \tilde Y = \left( p_\theta - q A_\theta \right)\basis[p_\theta] + \left( p_\varphi - q A_\varphi \right)\basis[p_\varphi] , \eeq is a Liouville vector field for $\omega_F$. On the energy level $\Sigma_v = \left\{ h = \tfrac{1}{2}mv^2 \right\}$ of the kinetic Hamiltonian~\eqref{eq.ham}, \beq \label{eq.transverse} \d h \left( \tilde Y \right) = m v^2 - \frac{q}{m}\, g^{-1}(p,A) , \qquad \min_{\Sigma_v} \d h\left(\tilde Y\right) = m v \left( v - \left\vert\alpha\right\vert \right), \eeq so $\tilde Y$ is transverse to $\Sigma_v$, and $\Sigma_v$ is of contact type for $\lambda'$, precisely when \beq \label{eq.threshold} v > \left\vert \alpha \right\vert = \frac{\left\vert qB \right\vert}{2m} . \eeq Each $\Sigma_v$ is diffeomorphic to the unit tangent bundle $T^1\Sph^2 \cong SO(3) \cong \Reals P^3$. \end{theo}

\begin{proof} Now, the canonical form is exact by construction and the twist is exact by Proposition~\ref{prop.exact}, so that $\omega_F = -\d\lcan + q\, \pi^*(\d A) = -\d\left( \lcan - q\, \pi^* A \right)$, which is the assertion~\eqref{eq.twistedexact}. For the Liouville property, let us contract the field~\eqref{eq.liouville} with the twisted form. The twist is pulled back from the base and $\tilde Y$ is vertical, so $\iiota_{\tilde Y}\pi^*(F) = 0$, while the canonical part gives $\iiota_{\tilde Y}\omega_0 = -\left(\lcan - q\pi^*A\right) = -\lambda'$. Hence \beq \label{eq.liouvilleproof} \pounds_{\tilde Y}\, \omega_F = \d\, \iiota_{\tilde Y}\, \omega_F = -\d \lambda' = \omega_F , \eeq by Cartan's formula and the closedness of $\omega_F$. Note that the fibrewise dilation about the zero section, obtained from~\eqref{eq.liouville} by deleting the potential, satisfies instead $\pounds_Y \omega_F = \omega_F - q\,\pi^*(F)$, and is therefore \emph{not} a Liouville field for the twisted form; the displacement of the centre of dilation from the zero section to the potential is what the twist demands.

For the transversality, the Hamiltonian~\eqref{eq.ham} is fibrewise quadratic, so Euler's identity gives $\d h(Y) = 2h = mv^2$ on $\Sigma_v$ for the undisplaced field, and the displacement contributes the term linear in the potential in~\eqref{eq.transverse}. Writing a covector of $\Sigma_v$ as $p = mv\left(\cos s\ \d\theta + \sin\theta \sin s\ \d\varphi\right)$, which is the general element of length $mv$, the potential~\eqref{eq.potential} gives \beq \label{eq.dhY} \d h\left(\tilde Y\right) = m v^2 - \frac{q v B}{2}\, \sin\theta\, \sin s . \eeq The second term of~\eqref{eq.dhY} attains $\mp \vert q\vert vB/2$ as $\sin\theta\sin s$ ranges over $[-1,1]$, and it does so on the equator with $p$ azimuthal, so the minimum over $\Sigma_v$ is $mv^2 - \vert q\vert vB/2 = mv(v-\vert\alpha\vert)$ by the abbreviation~\eqref{eq.alpha}, and it is attained. That minimum is positive exactly under the condition~\eqref{eq.threshold}, and $\d h(\tilde Y)$ then has no zero on $\Sigma_v$. Finally, $h$ is the fibrewise squared norm up to a constant, so $\Sigma_v$ is the bundle of covectors of fixed length, diffeomorphic to $T^1\Sph^2$ by the metric, and the unit tangent bundle of the round sphere is $SO(3)\cong \Reals P^3$. \end{proof}

\begin{remark}[The two thresholds coincide] \label{rem.threshold} The condition~\eqref{eq.threshold} is the condition $v > \vert\alpha\vert$ that Theorem~\ref{theo.ell0} identified as the boundary between the rotation and the libration regimes, and we arranged that agreement nowhere in the derivation of either. Indeed, the dynamical threshold arises from the turning points of the pendulum, and the contact threshold from the largest value the potential attains on the sphere. Both answer to the same ratio $\vert\alpha\vert/v$, the one Proposition~\ref{prop.scaling} isolates, because the potential attains its maximum exactly on the equator --- the very circle the trapped trajectories leave alone. The condition~\eqref{eq.threshold} is stated in $\vert\alpha\vert$, so the convention of Section~\ref{sec.closed} plays no part in it. Below the threshold the Liouville field~\eqref{eq.liouville} runs tangent to $\Sigma_v$ somewhere on that circle. No other Liouville field succeeds there either, since that criterion is an equivalence~\cite{contreras2004periodic} and Proposition~\ref{prop.mane} places those levels below $c_0$. \end{remark}

\begin{corollary} \label{cor.reeb} For $v > \vert\alpha\vert$ the Reeb vector field of the contact form $\lambda'\vert_{\Sigma_v}$ generates the electromagnetic flow of Definition~\ref{def.emcurve} reparametrised by arc length, and the closed trajectories of Theorem~\ref{theo.closure} --- those at the values of $\alpha$ solving the resonance~\eqref{eq.resonance} --- are closed orbits of that Reeb field, one on each level set whose rotation number is rational, up to the choice of starting point. \end{corollary}

\begin{proof} The Hamiltonian vector field $X_h$ of the kinetic Hamiltonian~\eqref{eq.ham} with respect to $\omega_F$ spans the characteristic line field of $\Sigma_v$, and so does the Reeb field $R$ of $\lambda'\vert_{\Sigma_v}$, since both are annihilated by $\omega_F\vert_{\Sigma_v}$. They are therefore proportional, and the two flows differ by a reparametrisation which is constant on $\Sigma_v$ because $h$ is fibrewise quadratic. Closed orbits correspond under a reparametrisation, and Theorem~\ref{theo.closure} supplies them for every $v > \vert\alpha\vert$, which is the very range on which Theorem~\ref{theo.contact} furnishes the contact form. \end{proof}

The existence results we cited in Section~\ref{sec.intro} are of Weinstein-conjecture flavour, and establish closed orbits on hypersurfaces of this kind without producing one. Here Theorem~\ref{theo.closure} produces them on $T^1\Sph^2 \cong \Reals P^3$ and names the values of $\alpha$ at which they occur, on the whole energy range where the contact structure is available. Therein lies what this example contributes.

%%============================================================
\section{Closing remarks} \label{sec.closing}
%%============================================================

We asked when a charge confined to a sphere retraces its path in a uniform ambient field. It closes when a resonance holds between the period of its polar oscillation and the rate of its azimuthal drift.

Let us set out what we have proved. The equations of motion follow from the Levi-Civita connection and the Lorentz endomorphism (Proposition~\ref{prop.eom}). The azimuthal Killing vector generates the Noether--Ikawa first integral $\ell$ (Proposition~\ref{prop.first}) which, with the speed, reduces the problem to a quadrature (Proposition~\ref{prop.quadrature}). A trajectory reaches a pole exactly on the level set $\ell = 0$ and is otherwise confined to a band (Proposition~\ref{prop.poles}). Over one polar period the azimuth advances by a complete elliptic integral of the third kind, which we reduce to Legendre form (Theorem~\ref{theo.rotation}). A trajectory closes exactly when that advance is a rational multiple of $2\pi$, after $q$ polar periods and $p$ azimuthal turns (Theorem~\ref{theo.closure}). That criterion is the manuscript's principal result, and it constrains one dimensionless ratio, since the rotation number depends on the half-cyclotron frequency, the azimuthal integral and the speed through $\alpha/v$ and $\ell/v$ alone (Proposition~\ref{prop.scaling}). The equator is a double root of the polar motion exactly when $(\ell+\alpha)^2 = v^2$, that is at the equatorial value $\alpha = v-\ell$, which carries a separatrix generalising the one of the boundary case when $\ell < v/2$ and the equatorial great circle otherwise (Proposition~\ref{prop.separatrix}). Below that value the rotation number rises strictly from the winding of a great circle (Theorems~\ref{theo.realise} and~\ref{theo.mono}), the derivative being the closed-form combination~\eqref{eq.twist} of the polar period and a second moment (Theorem~\ref{theo.twist}). Above it the band lies within one hemisphere, where an involution exchanging the turning points leaves the quadrature invariant (Lemma~\ref{lem.involution}) and the rotation number falls strictly (Theorem~\ref{theo.mono2}) to the winding $\sec\theta_c$ of a circular orbit or to zero (Theorem~\ref{theo.realise2}). Each rational winding is therefore carried by two values of $\alpha$, by one, or by none, and we say which (Corollary~\ref{cor.count}).

Moreover, on the boundary set $\ell = 0$ the motion is a Jacobi elliptic function of a single modulus and splits into rotation, libration and separatrix (Theorem~\ref{theo.ell0}), where the criterion becomes a resonance governed by a function strictly increasing onto the positive reals (Corollary~\ref{cor.ell0}). The field carries no flux and is globally exact (Proposition~\ref{prop.exact}), its Ma\~n\'e strict critical value is $\tfrac{1}{2}m\alpha^2$ (Proposition~\ref{prop.mane}), and above that value the energy levels are of contact type, with the Liouville field the fibrewise dilation displaced to the potential (Theorem~\ref{theo.contact}). The closed trajectories there are explicit closed Reeb orbits (Corollary~\ref{cor.reeb}).

The constitutive input is short. We declare each piece where we use it. We postulate the Lorentz force law in the form~\eqref{eq.eom}, with the sign convention of Definition~\ref{def.endo} and the twist~\eqref{eq.twisted} tied to it (Remark~\ref{rem.sign}). The field configuration is an input, the restriction of a uniform ambient field to a rigid sphere (Definition~\ref{def.restricted}). We model neither back-reaction nor radiation.

The selection rule is explicit. It carries a scale. Indeed, Remark~\ref{rem.resonant} names the ratio at which the simplest resonance closes, and with it the field a laboratory would need at a given charge, mass and speed, while the condition~\eqref{eq.resonance} fixes the higher ones. In this sense each rational winding selects one or two values of $\alpha$ out of a continuum, so a closed orbit stands isolated among its precessing neighbours, while the closed orbits together occupy a set of $\alpha$ that is dense and of Lebesgue measure zero.

We compute the rotation number of this problem in closed form, we prove the closure criterion it yields, we count the values of $\alpha$ that meet it, and we show that the contact threshold coincides with the dynamical one (Remark~\ref{rem.threshold}). Section~\ref{sec.intro} names what we build on --- the integrability of this very flow~\cite{dragovic2023gyroscopic,bolsinov2025integrability}, Sternberg's lift~\cite{sternberg1977minimal,guillemin1990symplectic}, Ikawa's extension of Noether's theorem~\cite{ikawa2003hamiltonian,ikawa2007motion,ikawa2010motion}, the definition of the Lorentz endomorphism and its application to a monopole~\cite{lopez2025noether}, the characterisation of Gauss--Landau--Hall curves as those of constant geodesic curvature~\cite{barros2005gauss}, and the existence of closed magnetic geodesics on the two-sphere~\cite{ginzburg2007periodic,usher2008floer,schneider2009alexandrov,asselle2016existence}.

One matter remains open. It concerns the trajectories below the threshold~\eqref{eq.threshold}. Those levels fail to be of contact type, by the equivalence of Contreras, Macarini and Paternain~\cite{contreras2004periodic}, yet Theorem~\ref{theo.closure} still supplies closed orbits among them. What those orbits are, absent a Reeb field to generate them, is the question about sub-critical dynamics at which the construction stops.

Three directions lead out of those boundaries, and we name them without pursuing them. The comparison with the monopole can be made structural rather than methodological, by asking which features of a first integral survive the passage from an exact to a non-exact field. The construction is not tied to the sphere. The other constant-curvature surfaces carry restricted ambient fields of their own. Finally, the premetric formulation of electromagnetism~\cite{hehl2003foundations}, in which the field is prior to the metric, would let the constitutive relation vary while the topology stays fixed. Each would be a manuscript of its own.

\bibliographystyle{plain} \bibliography{em_curves_sphere}

\end{document}